\documentclass{llncs}

\usepackage{graphicx} %
\usepackage{hyperref} %
\usepackage[firstpage]{draftwatermark} %

\SetWatermarkAngle{0}
\SetWatermarkText{\raisebox{12.5cm}{
\hspace{0cm}
\href{https://doi.org/10.5281/zenodo.15309251}{\includegraphics{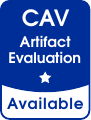}}
\hspace{12.2cm}
\includegraphics{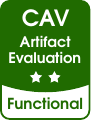}
}}

\usepackage{graphicx} %
\usepackage{amsmath, amssymb}
\usepackage{thmtools}
\usepackage{hyperref}
\usepackage{bbm}
\usepackage[ruled,linesnumbered,vlined,noend]{algorithm2e}
\usepackage{amsfonts}
\usepackage{float}
\usepackage{listings}
\usepackage{caption}
\usepackage{tikz}
\usetikzlibrary{calc}
\usetikzlibrary{tikzmark}
\usepackage{subcaption}
\usepackage{todonotes}
\usepackage{cleveref}
\usepackage{enumitem}
\usepackage{lineno}
\lstdefinestyle{compactpython}{
  language=Python,
  basicstyle=\ttfamily\footnotesize,
  keywordstyle=\color{blue}\bfseries,
  commentstyle=\color{gray}\itshape,
  stringstyle=\color{red},
  numbers=none,
  breaklines=true,
  frame=none,
  backgroundcolor=\color{gray!10},
  xleftmargin=0pt,
  xrightmargin=0pt,
  columns=flexible,
}

\usepackage{xcolor}
\hypersetup{
    colorlinks=true,
    allcolors=black
}
\usepackage{orcidlink}

\definecolor{color1}{HTML}{F6BD60}
\definecolor{color2}{HTML}{F7EDE2}
\definecolor{color3}{HTML}{F5CAC3}
\definecolor{color4}{HTML}{84A59D}
\definecolor{color5}{HTML}{F28482}

\lstdefinestyle{custompython}{
  language=Python,
  basicstyle=\ttfamily\small,
  keywordstyle=\bfseries\color{blue},
  commentstyle=\itshape\color{gray},
  stringstyle=\color{green!50!black},
  showstringspaces=false,
  frame=single,
  breaklines=true,
  backgroundcolor=\color{lightgray!20},
  captionpos=b,
  tabsize=4,
  morekeywords={from, str, import},
  keywordstyle=[2]\color{red!50!black}, 
  morekeywords=[2]{binomial},
  keywordstyle=[3]\color{red}, 
  morekeywords=[3]{smartbinom}
}

\def\HiLi{\leavevmode\rlap{\hbox to \hsize{\color{gray!20}\leaders\hrule height .8\baselineskip depth .5ex\hfill}}}

\newcommand{\eps}{\varepsilon}

\newcommand{\E}{\mathop{\mathbb{E}}}
\newcommand{\remove}[1]{}

\newcommand{\cA}{\mathcal{A}}

\newcommand{\Oh}{\mathcal{O}}

\newcommand{\bF}{\mathbb{F}}
\newcommand{\bR}{\mathbb{R}}

\newcommand{\sOh}{o}

\newcommand{\cX}{\mathcal{X}}

\newcommand{\rnd}{\mathtt{rnd}}

\newcommand{\dtv}{\mathsf{d_{TV}}}

\newcommand{\fail}{\textsf{Fail}}
\newcommand{\bad}{\mathsf{Bad}}

\newcommand{\sampler}{\ensuremath{\mathsf{BinSamp}}}

\newcommand{\clog}{\ensuremath{\mathsf{Log}}}

\newcommand{\logfact}{\ensuremath{\mathsf{LogFactorial}}}

\newcommand{\lancz}{\ensuremath{\mathsf{Fact}^\mathsf{Lancz}}}

\newcommand{\aps}{\texttt{APSEst}}
\newcommand{\apsnew}{\texttt{APSEst2}}
\newcommand{\ganak}{\texttt{Ganak}}

\newcommand{\epsh}{\eps_{\dchat}}
\newcommand{\dhat}{\mathsf{h}}
\newcommand{\dinvhat}{\mathcal{H}^{-1}}
\newcommand{\dchat}{\mathcal{H}}

\newcommand{\stirlrs}{\mathcal{B}^{\mathtt{Stirl}}_{rej}}

\newcommand{\fplus}{\oplus}
\newcommand{\fminus}{\ominus}
\newcommand{\ftimes}{\otimes}
\newcommand{\hcirc}{\circledast}

\newcommand{\accept}{\mathsf{accept}}

\newcommand{\bin}{{\mathsf{b}_{n,p}}}
\newcommand{\sampdist}{\mathsf{b}^{\mathsf{BinSamp}}_{n,p}}
\newcommand{\interdist}{\mathsf{b}^{\mathsf{Inter}}_{n,p}}
\newcommand{\barbin}{{\overline{\mathsf{b}}_{n,p}}}
\newcommand{\p}{\mathsf{p}}
\newcommand{\q}{\mathsf{q}}
\newcommand{\bb}{\mathsf{b}}
\newcommand{\erra}{\texttt{E1}\xspace}
\newcommand{\errb}{\texttt{E2}\xspace}
\newcommand{\rinter}{r^{\mathsf{Inter}}_k}
\newcommand{\crdy}[1]{{\color{black} #1}}

\begin{document}

\title{Assessing the Quality of Binomial Samplers: \\ A Statistical Distance Framework}

\author{
    Uddalok Sarkar\orcidlink{0009-0000-4997-1084}\inst{1} \and    
    Sourav Chakraborty\orcidlink{0000-0001-9518-6204}\inst{1} \and
    Kuldeep S. Meel\orcidlink{0000-0001-9423-5270}\inst{2}
}

\institute{
    Indian Statistical Institute, India 
    \and
    Georgia Institute of Technology, USA and University of Toronto, Canada 
}

\maketitle

\begin{abstract}
    Randomized algorithms depend on accurate sampling from probability distributions, as their correctness and performance hinge on the quality of the generated samples. However, even for common distributions like Binomial, exact sampling is computationally challenging, leading standard library implementations to rely on heuristics. These heuristics, while efficient, suffer from approximation and system representation errors, causing deviations from the ideal distribution. Although seemingly minor, such deviations can accumulate in downstream applications requiring large-scale sampling, potentially undermining algorithmic guarantees. In this work, we propose statistical distance as a robust metric for analyzing the quality of Binomial samplers, quantifying deviations from the ideal distribution. We derive rigorous bounds on the statistical distance for standard implementations and demonstrate the practical utility of our framework by enhancing APSEst, a DNF model counter, with improved reliability and error guarantees. To support practical adoption, we propose an interface extension that allows users to control and monitor statistical distance via explicit input/output parameters. Our findings emphasize the critical need for thorough and systematic error analysis in sampler design. As the first work to focus exclusively on Binomial samplers, our approach lays the groundwork for extending rigorous analysis to other common distributions, opening avenues for more robust and reliable randomized algorithms.
\end{abstract}

\section{Introduction}\label{sec:introduction}

Randomization stands as a cornerstone of computer science, permeating algorithm design from the field's earliest days to its cutting-edge developments. From Quicksort~\cite{hoare1961algorithm1}, one of the most widely used algorithms, to modern cryptographic protocols, randomization has proven indispensable in achieving efficiency and functionality that deterministic approaches struggle to match. While the fundamental question of whether randomization offers additional computational power over determinism remains open, randomized algorithms have established themselves as the preferred choice in numerous domains, including data structures~\cite{pugh1990concurrent}, hash functions~\cite{carter1977universal}, and probabilistic data structures~\cite{bloom1970space}.

At the heart of every randomized algorithm lies its ability to sample from probability distributions. The algorithm's correctness and performance guarantees intrinsically depend on the quality of these samples. For instance, a hash table's performance relies on the uniformity of its hash function's output distribution, while a Monte Carlo algorithm's accuracy depends on the fidelity of its random sampling process. This fundamental reliance on sampling has led to the development of sophisticated sampling algorithms implemented as standard library functions across programming languages.

While specialized techniques exist for generating high-quality samples from certain 
distributions~\cite{karney2016sampling}, these approaches typically circumvent direct 
probability mass computation through transformations of basic random processes. However, 
such techniques remain constrained to specific distributions exhibiting particular 
mathematical properties. In practice, standard library implementations predominantly rely 
on transformed rejection sampling~\cite{hormann1993generation,hormann2004transformed}, 
which necessitates explicit probability mass computation. These computations entail 
multiple arithmetic operations and specialized function evaluations, including factorial 
and logarithm computations, thereby introducing approximation errors at each step. The 
accumulation of these errors can significantly impact the statistical properties of the 
generated samples, potentially compromising the theoretical guarantees of algorithms 
that depend on them.

In this work, we focus on analyzing standard library implementations of Binomial 
samplers, which are largely based on transformed rejection sampling 
techniques~\cite{hormann1993generation,ormann1994transformed,hormann2004transformed}. 
These implementations require computation of Binomial distribution probability mass, denoted by 
$\bin(k)$, necessitating approximations of factorial terms~\cite{lanczos1964precision}, 
logarithmic computations~\cite{borwein1984arithmetic}, and various arithmetic 
operations. While such approximations enable efficient sampling, they introduce 
systematic deviations from the ideal Binomial distribution that current implementations 
neither quantify nor report to users. These deviations can accumulate and potentially 
trigger cascading failures in downstream applications~\cite{binder1992monte,thomopoulos2012essentials}. Despite the widespread adoption of these libraries, there 
exists no documentation providing precise analysis of accumulated errors.

The primary research problem we address is: \emph{how to develop a rigorous 
	methodology to analyze the errors in existing samplers to provide meaningful 
	measurement of their impact on downstream applications?} This question is particularly 
pertinent given the increasing reliance on randomized algorithms in critical 
applications, where understanding and quantifying sampling errors becomes crucial for 
ensuring system reliability and correctness.

Our first contribution is a rigorous framework for analyzing the quality of existing samplers through the lens of statistical distance. We advocate statistical distance as a theoretically sound metric for quantifying sampler quality, owing to its fundamental property of indistinguishability. Let $\p$ and $\q$ be two probability distributions with statistical distance at most $\eta$, i.e., $\dtv(\p,\q) \leq \eta$. Then, for any statistical test $T$ (even computationally unbounded), the probability of distinguishing between samples from $\p$ and $\q$ is bounded by $\eta$. This fundamental property has profound implications for sampler quality analysis: if a sampler's output distribution has a statistical distance $\eta$ from the ideal distribution, then the downstream application cannot experience an error greater than $\eta$, regardless of its computational sophistication. Building on this theoretical foundation, we present a detailed analysis of state-of-the-art implementations, deriving concrete bounds on their deviation from the ideal distributions through careful decomposition of numerical approximation errors. We propose an extension to sampler interfaces that exposes statistical distance as an input/output parameter, enabling users to control and monitor the sampling accuracy. 

To demonstrate the practical utility of our theoretical framework, we present a comprehensive case study in the context of DNF model counting. We show how our quality measures can be integrated into \aps{}, a state-of-the-art DNF counter that relies heavily on Binomial sampling. By incorporating our error bounds into the \aps{}'s analysis framework, we provide the first implementation that offers both scalability and rigorous error guarantees. This integration not only enhances the reliability of the counter but also establishes a blueprint for how sampler quality analysis can be systematically incorporated into broader algorithmic frameworks. Our empirical evaluation demonstrates that this enhanced implementation maintains the efficiency of existing approaches while providing substantially stronger theoretical guarantees.

We believe our work highlights a fundamental challenge in randomized computation: the need for rigorous analysis of sampler implementations to establish precise error bounds and enhance trust in randomized algorithms. While we have focused on Binomial samplers as a crucial first step, the theoretical framework we develop for analyzing sampling error propagation, combined with our practical demonstration in DNF counting, establishes a foundation for future research. A natural direction for future investigation would be the analysis of other standard distributions such as Poisson, Normal, and Beta distributions, each presenting its own unique challenges in implementation and error analysis. Our approach of integrating error analysis into algorithmic frameworks opens new avenues for developing robust randomized algorithms that maintain both theoretical guarantees and practical efficiency. This work will likely motivate the broader community to examine and enhance the reliability of randomized computation implementations, particularly in the context of standard library functions that serve as building blocks for numerous algorithms.

\paragraph*{Organisation} 
In section~\ref{sec:prelim}, we present the necessary preliminaries and an overview of related concepts that lay the foundation for the rest of the paper. In section~\ref{sec:related}, we explore related work on error analysis in computational programs and the evaluation of sampler quality. Section~\ref{sec:proposal} outlines our proposal for using statistical distance as a quality metric for samplers, along with the motivation behind this approach.
Section~\ref{sec:binosam} offers a detailed discussion on standard Binomial samplers, including our theoretical results, correctness proofs, and error analysis. In section~\ref{sec:casestudy}, we include a case study on using our sampler to count the number of solutions of a Boolean formula in the Disjunctive Normal Form. Finally, in section~\ref{sec:conclusion}, we discuss the limitations of our work and future directions.
\section{Preliminaries}
\label{sec:prelim}

In this work, we are interested in probability distributions over discrete sets and their samplers. A probability distribution, or simply, a distribution (denoted by $\p$) on a discrete set $\Omega$ is a mapping $\p:\Omega \to [0,1]$ such that $\sum_{x \in \Omega} \p(x) = 1$. We define $\p(A) = \sum_{x \in A} \p (x)$ for any $A \subseteq \Omega$. \crdy{A uniform distribution, or uniform randomness over a set $\Omega$ is defined as $\p(x) = \frac{1}{|\Omega|}$ for all $x \in \Omega$. We use the notation $\bin$ to denote the Binomial distribution with parameters $n$ and $p$, which is given by $\bin(k) = \binom{n}{k} p^k (1-p)^{n-k}$ for $k \in [0,n]$.}

\crdy{
	Recall that a Turing Machine (TM) is a theoretical model of computation defined as a tuple \( (Q, \Sigma, \Gamma, \sqcup, \Delta, s_0, F) \), where \( Q \) is a finite set of states, \( \Sigma \subseteq \Gamma \setminus \{\sqcup\} \) is the input alphabet, \( \Gamma \) is the tape alphabet, \( \sqcup \in \Gamma \) is the blank symbol, \( \Delta: Q \times \Gamma \to Q \times \Gamma \times \{L, R\} \) is the transition function, \( s_0 \in Q \) is the initial state, and \( F \subseteq Q \) is the set of final states~\cite{hopcroft2001introduction}. A natural extension of a Turing Machine is a Turing Transducer~\cite{meduna2000turing}, which, in addition to the input and work tapes, has a separate write-only output tape. A Transducer computes a function \( f: \{0,1\}^* \to \{0,1\}^* \), and the output is the content of the output tape when the machine halts.
	A Probabilistic Turing Machine (PTM) is a Turing Machine that, in addition to the input tape, has access to a read-only random tape filled with an infinite sequence of random bits~\cite{arora2009computational}. On a given input \( x \in \{0,1\}^* \) and for each fixed random string \( u \in \{0,1\}^\infty \), the machine's behavior is deterministic. A probabilistic Transducer is defined analogously as a PTM equipped with an output tape. It computes an output string \( M(x; u) \) for each fixed \( u \), and writes it on the output tape. 

	A randomized algorithm \( \mathcal{A} \) is modeled as a probabilistic Transducer. On input \( x \) and a source of randomness, the output of the algorithm \( \mathcal{A}(x; u) \) is written on the output tape of the corresponding Transducer. Consequently, $\cA$ defines a distribution over outputs depending on the randomness. 
	While this definition assumes that the random bits are drawn from the uniform distribution, we allow randomized algorithms to access randomness from arbitrary distributions. 
	The ability to sample from arbitrary distributions is without loss of generality, since any distribution can be simulated using uniformly random bits.
}
An example of a randomized algoithm is a sampler $\mathsf{Samp}_\p$ for a distribution $\p$. Given a source of uniform randomness $u$, the sampler outputs a sample from $\p$, that is, for all $x\in \Omega$, ${\Pr_u}(\mathsf{Samp}_\p \text{ outputs } x) = \p(x)$. Conversely, a sampler induces an associated probability distribution $\p$ from which it draws samples.  

\subsection{Approximating Factorials}
\label{sec:factapprox}

Lanczos Approximation~\cite{lanczos1964precision} is a widely used method to approximate the factorials with remarkable accuracy. For a fixed value of $t,g$, and a positive integer $n\geq 1$, the Lanczos approximation of $n!$, denoted by $\lancz(n)$, is defined as, 
$\lancz(n) = \sqrt{2 \pi}\left(n+g+\frac{1}{2}\right)^{n+\frac{1}{2}} e^{-(n+g+\frac{1}{2})} A_{t,g}(n)$.
The polynomial $A_{t,g}(n)$ contains $t$ terms. The accuracy of the approximation depends on the number of terms $t$ in its expansion, as well as on the constant $g$. Here $g$ is any real constant such that \crdy{$g + \frac{1}{2} > 0$}. 
The parameters $g$ and $t$ affect the accuracy and convergence rate of the Lanczos approximation, where larger $t$ improves accuracy at the cost of higher computational resources.

In the Lanczos approximation, a uniform error bound~\cite{pugh2004analysis} can be established, which provides a measure of how closely the Lanczos approximation approximates the factorial function for all relevant inputs.
Given $t,g$ the uniform error bound $\zeta_{t,g}$ of the approximation is defined by,
$\zeta_{t,g} = \sup_{n \in \mathbb{N}} \left|n! - \lancz(n)\right|$.
Let $\zeta = \sqrt{\frac{\pi}{e}} \cdot |\zeta_{t,g}|$. The relative error can be bounded as follows~\cite{pugh2004analysis}, 
\begin{equation}\label{eq:lanczerr}
	\frac{|n! - \lancz(n)|}{n!} \leq \zeta
\end{equation}

\subsection{Multiple-precision Arithmetic}
\label{sec:mparithmetic}

Given a working precision $\beta > 0$, the set of all definable numbers in this context is expressed as $\bF = \left\{ w \cdot 2^{e} \ : \ \frac{1}{2} \leq |w| \leq 1 \text{ and } e\in \mathbb{Z}\right\}$~\cite{mpfralgorithms,muller2006elementary,jeannerod2018relative}.
Here $e$ is an integer denoting the exponent, and the $\text{ulp}(x) = 2^{e-\beta}$, where ulp denotes the unit in the last place~\cite{mpfralgorithms}. Let $\rnd : \bR \to \bF$ be the rounding function that rounds a real number to the nearest definable number. The corresponding relative errors are bounded by $\frac{|\rnd(x) - x|}{|x|} \leq \eps$, for $x \neq 0$, where $\eps = \frac{1}{2^{\beta}}$, is referred to as the \textit{unit round-off}~\cite{mpfralgorithms,muller2006elementary,jeannerod2018relative}.

We define a set of operations by \textit{basic operations} for which it is possible to directly compute the correct rounding of the result~\cite{fousse2007mpfr}. These operations are $\{+,-,\times,/,\sqrt{}\}$. For any two numbers $x,y \in \bR$, the following bound holds:
\begin{equation}
	\label{eq:basicops}
	\left|(\rnd(x) \hcirc \rnd(y)) - (x \ast y)\right| \leq \eps \cdot |x \ast y|	
\end{equation}
where $\ast \in \{+,-,\times,/\}$, and, $\hcirc$ is the corresponding operation in $\bF$. Same bound holds for $\sqrt{x}$~\cite{jeannerod2018relative}. 
For $n$ real numbers $x_1, x_2, \ldots, x_n$, the computed sum $\hat{s} := \rnd(x_1) \fplus \rnd(x_2) \fplus \ldots \fplus \rnd(x_n)$, regardless of the order of computation, deviates from the exact sum $s = \sum_{i=1}^n x_i$ by at most following bound~\cite{jeannerod2018relative},
\begin{equation}
	\label{eq:basicops-sum}
	\left|\hat{s} - s\right| \leq n\eps \sum_{i=1}^{n}|x_i|
\end{equation} 

The basic operations are the building blocks for other advanced operations, such as logarithms, exponentials, and trigonometric functions. 

\subsection{Approximate Computation of Logarithm}
Logarithm computation is generally approximated using the Taylor series. However, for high precision, \textit{arithmetic-geometric-mean (AGM)}~\cite{borwein1984arithmetic} is used. 
Let us consider two sequences $\{w_n\}, \{z_n\}$ of positive real numbers such that, $w_{n+1} = \frac{w_n + z_n}{2}, z_{n+1} = \sqrt{w_n \cdot z_n}$.
These two sequences converge to the common limit and are denoted by $AGM(w_0, z_0)$. 

\crdy{
For $x \in \bF$ represented as $x = w \cdot 2^e$, we define the function $\mathrm{exponent}(x) =e$. To compute $\log(x)$ using the $AGM$ method, an integer $m$ is computed such that $x2^m > 2^{\beta/2}$. The algorithm then evaluates $AGM(1, 4/s)$ and computes the logarithm as \[\clog(x) = \frac{\pi}{2AGM(1, 4/s)} - m\log 2\] 
In the MPFR library~\cite{mpfralgorithms}, the value of $m$ is fixed as $m = \left\lceil \frac{\beta + 3}{2} \right\rceil - \mathrm{exponent}(x)$. This choice of \( m \) ensures that \( s = x 2^m \) lies within the range \([2^{\beta/2}, 2^\beta]\).}
The following lemma provides an error bound for the AGM method. 

\begin{lemma}[\crdy{Prop. 2 of \cite{borwein1984arithmetic}}]
	\label{lem:elliptic}
	For the function $AGM$, the following holds for any $s \geq 4$:
	$\left|\frac{\pi}{2AGM(1, 4/s)} - \log\left(s\right)\right| \leq  \frac{64}{s^2}(10 + |\log s|)$.
\end{lemma}

Let $\clog$ be the function that computes the logarithm using the AGM method.
Since, for $\beta > 8$, we have $3\log(s) > 10$ and $2^{\beta/2} \leq s \leq 2^\beta$, we make the following conclusion:
\begin{equation}
	\label{eq:approxlog}
	\left|\clog(x) - \log(x)\right| \leq  \frac{178\beta}{2^\beta}
\end{equation}

We will use the notation $\tau = \frac{178\beta}{2^\beta}$ as the additive error bound for the logarithm approximation in the rest of the paper. Note that, $\tau= \Oh(\eps)$. A similar error bound can be derived for the Taylor series method as well~\cite{bonnot2023formally,bonnot2024formally}. 

Evaluating the logarithm of the factorial, rather than the factorial itself, is the standard technique. The function $\logfact$ computes the logarithm of the factorial using the Lanczos approximation with fixed parameters $t$ and $g$.
\begin{align}\label{alg:logfact}
	\logfact(k) &= \frac{1}{2} \clog{}{}(2\pi) + \left(k + \frac{1}{2}\right) \clog{}{} \left(k + g + \frac{1}{2}\right) \notag \\
	&\quad\quad\quad\quad\quad\quad\quad - \left(k + g + \frac{1}{2}\right) + \clog \left(A_{t,g}(k)\right)	
\end{align}

\section{Related Work}
\label{sec:related}

The impact of computational approximations has been a longstanding concern in the literature. Considerable effort has been devoted to designing samplers that generate samples with arbitrary precision and provably no deviation from the original distribution, a concept referred to as \textit{exact sampling}. This line of work dates back to Von Neumann and has been further developed in studies such as \cite{karney2016sampling}, which propose arbitrarily precise algorithms for sampling from distributions like the normal and exponential. The core idea involves employing a random process that efficiently generates a sample \( x \) with probability \( e^{-x} \). Remarkably, this algorithm achieves an expected runtime of \( \Oh(1) \). 

Similarly, significant attention has been given to designing exact Binomial samplers \cite{devroye1980generating,farach2015exact} as well. The approach in \cite{devroye1980generating} employs the geometric distribution to generate Binomial samples but requires \( \Oh(np) \) time. More recent advances by \cite{farach2015exact} achieve $\Oh(\sqrt{n})$ time complexity. Their approach involves efficiently accessing \(\mathsf{b}_{n, 1/2}\) and leveraging its samples, combined with the binary representation of \( p \), to generate samples from \(\mathsf{b}_{n, p}\).

Constant time sampling algorithms for binomial distributions are categorized under the framework of \textit{transformed rejection sampling}~\cite{devroye2006nonuniform,schmeiser1981poisson,hormann1993generation,kachitvichyanukul1988Binomial}. 
These algorithms achieve a sampling time complexity of $\Oh(1)$, but at the cost of approximations. This is because the framework needs to evaluate the probability mass function, which is computationally expensive unless approximated.

The impact of numerical accuracy on computational programs has been extensively studied. Significant research has been conducted to analyze errors in arithmetic operations~\cite{zimmermann2006error,jeannerod2013improved,jeannerod2017error,jeannerod2018relative,rump2008accurate}. . 
Recently, \cite{blanchard2021accurately} and \cite{bonnot2024formally} have explored how these errors affect the performance of functions such as log-sum-exp and softmax. These studies underscore the critical need to account for the inherent numerical errors when designing algorithms and assessing their practical performance. 

Finally, statistical distance has been widely recognized as a key measure of sampler quality. For instance, a series of works~\cite{chakraborty2019testing,meel2020testing,PM21,PM22,pmlr-v206-banerjee23a,kumar2023tolerant,bhattacharyya2024testing} focus on designing tests to determine the quality of samplers in terms of the statistical distance between the sampler and the target distribution. 
\section{Statistical Distance as Quality Metric}\label{sec:proposal}
Since exact sampling from distributions such as Binomial is computationally expensive for most parameters of interest, the standard libraries rely on approximations to achieve practical efficiency. While these approximations significantly reduce time complexity, they introduce deviation from the actual distribution, effectively causing the samples to come from a distribution different from the intended one. Therefore, we need to focus on a fundamental question: {\em how do we make systems that rely on samplers trustworthy?}

Simply ignoring these deviations is not advisable, as they can have cascading effects that compromise the correctness of the entire system. Often, a user designs a randomized algorithm $\cA$ to solve a particular problem, with an upper bound $\delta$ on its failure probability. If $\cA$ relies on a standard Binomial sampler without knowledge of the sampler's quality, the program may experience a higher failure rate due to approximations in the underlying samplers. 

Our proposal immediately raises the question: how should one measure the quality of the sampler? To this end, we first focus on the fact that the objective of the measurement of quality is to allow the end user to quantify the impact of the usage of the sampler. There are several metrics, such as KL-divergence, statistical distance, and Hellinger distance, that have been proposed in the literature focused on probability distributions that seek to quantify the distance between two probability distributions. In this regard, a natural question is to ask what distance metric we should choose. To this end, we propose the usage of statistical distance (\Cref{def:tvdistance}) as the metric to report the quality. 

\begin{definition}[Statistical Distance]\label{def:tvdistance}
	Suppose two distributions $\p$, $\q$ are defined over the set $\Omega$. The Statistical Distance (denoted by $\dtv$) between $\p$ and $\q$ is defined by, $ \dtv(\p, \q) =  \frac{1}{2}\sum_{x \in \Omega}\left|\p(x) - \q(x)\right| = \sup_{A \subseteq \Omega}\p(A) - \q(A)$.
\end{definition}

Our proposal for statistical distance stems from its ability to allow end users to derive the worst-case bounds on the behavior of the system in a {\em black-box} manner. Formally, this follows from the folklore lemma below, for which we provide a proof for completeness.

\begin{lemma}
	\label{lem:indistinguishable}
	Let $\cA$ be a randomized algorithm \crdy{that uses randomness from a source} distribution $\p$, and let $\mathsf{Bad}$ be an event in the output of $\cA$. If $\p$ is replaced by another distribution $\q$, then the probability of the event $\mathsf{Bad}$ is bounded by the statistical distance between:
	\begin{align*}
		\left|\Pr_{r\sim\p}(\cA(x;r) \in \mathsf{Bad}) - \Pr_{r\sim\q}(\cA(x;r) \in \mathsf{Bad})\right| \leq \dtv(\p, \q)
	\end{align*}
\end{lemma}
\begin{proof}
	\crdy{
	Let $B \subseteq \Omega$ be the set of random strings (or, numbers) that trigger the event $\mathsf{Bad}$. Then, we have $\left|\Pr\limits_{r\sim\p}(\cA(x;r) \in \mathsf{Bad}) - \Pr\limits_{r\sim\q}(\cA(x;r) \in \mathsf{Bad})\right|= |\p(B) - \q(B)|$. Using the definition of statistical distance (\cref{def:tvdistance}), $|\p(B) - \q(B)| \leq \sup_{A \subseteq \Omega} \p(A) - \q(A) = \dtv(\p, \q)$.\qed
	}
\end{proof}

Note that the lemma above imposes no restrictions on $\cA$ or the event $\mathsf{Bad}$, highlighting the power of statistical distance as a metric. In particular, if $\dtv(\p, \q)$ is small, then the end user can be confident in 
bounding the overall impact on the program. We give a general recipe of how to incorporate statistical distance in the implementation of randomized algorithms.

\subsection{Integrating Statistical Distance Analysis in Applications}
\label{subsec:modifying-algorithms}

The correctness of randomized algorithms typically relies on access to exact samples from 
a target distribution $\p$. However, in practice, algorithms must use samplers that draw 
from an approximate distribution $\q$, potentially compromising their theoretical 
guarantees. We propose a systematic framework for incorporating these approximations while 
maintaining rigorous error bounds through minimal modifications to existing algorithms and 
their analyses.

\paragraph{Algorithm Modification}
Let $\cA$ be a randomized algorithm that requires samples from distribution $\p$. We 
modify $\cA$ to explicitly track and bound the accumulated error from using an 
approximate sampler as follows:

\begin{enumerate}
	\item Introduce an error budget parameter $\delta_1$ representing the maximum 
	allowable error due to sampling approximations.
	
	\item Initialize an error accumulator $\delta'$ to track the statistical distance:
	\begin{align*}
		\delta' \gets 0.
	\end{align*}
	
	\item For each sampler invocation, update the accumulated error:
	\begin{align*}
		\delta' \gets \delta' + \dtv(\p, \q)
	\end{align*}
	where $\dtv(\p, \q)$ is the pre-computed statistical distance bound.
	
	\item If $\delta'$ exceeds the budget ($\delta' > \delta_1$), abort execution.
\end{enumerate}

\paragraph{Analysis Modification}
Let $\delta_2$ denote the original error probability of algorithm $\cA$ assuming access 
to exact samples from $\p$. After incorporating the sampling approximation error 
$\delta_1$, the total error probability $\delta$ is bounded by: $$\delta \leq \delta_1 + \delta_2$$

This framework maintains theoretical guarantees while transparently accounting for 
sampling approximations. The modifications are minimal and the analysis remains 
straightforward. We demonstrate an end-to-end integration of this approach through a case 
study in \cref{sec:casestudy}.

An alternative approach would be to directly analyze algorithm $\cA$ with respect to the 
approximate distribution $\q$. However, this presents several challenges. The target 
distribution $\p$ often possesses mathematically convenient properties that facilitate 
analysis, while the implementation-specific $\q$ may lack such properties, making direct 
analysis intractable. Furthermore, updates to the underlying sampler implementation would inevitably necessitate re-analysis of every dependent algorithm.

Our framework enables separation of concerns: algorithm designers can conduct analysis 
assuming access to the idealized distribution $\p$, while library developers focus on 
bounding the statistical distance between $\p$ and their implementation $\q$. The 
errors can then be composed as shown above, providing rigorous bounds with minimal 
modification to existing analyses. 

\subsection{Proposal for Extending Sampler Interfaces}
\label{subsec:extending-sampler-api}

To enable seamless integration of our statistical distance framework, we propose extending the interface of existing samplers by incorporating two new components (see \cref{fig:sampler-interface}): (1) an input parameter $\delta_{in}$ that allows users to specify the maximum allowable statistical distance from the ideal distribution, and (2) an output parameter $\delta_{out}$ that reports the actual statistical distance achieved during sampling. 
These additions give users fine-grained control over the sampler. By setting $\delta_{in}$, users can explicitly define their tolerance for the statistical distance, while $\delta_{out}$ enables real-time monitoring of the sampler's performance. 
The fine-grained control offered by our interface allows users to make informed decisions about the trade-off between accuracy and performance. 

\begin{figure}[t]
\centering

\begin{minipage}[t]{0.45\textwidth}
\lstset{style=compactpython}
\begin{lstlisting}
def BinSamp(n, p):
    """
    Input: 	n, p
    Output:	sample
    """
    ...
\end{lstlisting}
\end{minipage}
\hspace{0.05\textwidth}
\begin{minipage}[t]{0.45\textwidth}
\lstset{style=compactpython}
\begin{lstlisting}
def BinSamp(n, p, delta_in):
    """
    Input:	n, p, delta_in
    Output:	sample, delta_out 
    """
    ...
\end{lstlisting}
\end{minipage}

\caption{Early (left) and new (right) sampler interfaces. The new version includes statistical distance control via \texttt{delta\_in} and \texttt{delta\_out}.}
\label{fig:sampler-interface}
\end{figure}

\section{Analysis of Standard Binomial Samplers}
\label{sec:binosam}
This section examines standard Binomial sampling algorithms and their inherent errors. We first present the standard Binomial sampling algorithm and then analyze the bounds on the statistical distance between the actual distribution and the distribution from which the sampler draws the samples.

\subsection{Standard Binomial Sampling Algorithms}

This section describes the standard Binomial sampling algorithms commonly used in practice, with particular attention given to Python's implementation.
These samplers rely on the method of \textit{transformed rejection sampling} which combines two well-established sampling techniques: (1) inverse transform sampling and (2) rejection sampling. Rejection sampling requires existence of an, \textit{easy to sample from}, hat distribution $\dhat$ such that for all $k\in [n]$, $\bin(k) < \alpha \dhat(k)$ for some $\alpha > 0$, known as \textit{rejection rate}. Inverse transform sampling generates samples from $\dhat$. Suppose the cumulative distribution of $\dhat$ is denoted by $\dchat$. Because $\dchat$ is a cumulative distribution therefore, its inverse $\dinvhat$ is well defined. To get a sample from $\dhat$, a uniform random variable $u$ is generated, and correspondingly the sample $k = \lfloor\dinvhat(u)\rfloor$ is computed. From the principles of inverse transform sampling, we can show that $k \sim \dhat$. The next step involves rejection sampling. After generating $k$, another uniform random sample $v$ is generated from $[0,1]$. The sample is rejected if $v > \frac{\bin(k)}{\alpha \dhat(k)}$, else $k$ is returned. 

Among the various Binomial sampling algorithms, the choice of the (inverse) hat distribution  $\dinvhat(u)$ is a key difference. For example, H{\"o}rmann~\cite{hormann1993generation} considered the following definitions of hat distribution\footnote{Since \((\dinvhat)'(u) = \frac{1}{\dhat(k)}\), we use the notation \(\dhat^{-1}(u) = \frac{1}{\dhat(k)}\) directly for simplicity.} for $-0.5\leq u\leq 0.5$, which has high acceptance probabilities for Binomial distributions over varied $n,p$. 

\begin{equation*}
	\label{eq:hat}
	\dinvhat(u) = \left(\frac{2\lambda_{n,p}}{(1/2 - |u|)} + \mu_{n,p}\right)u + \nu_{n,p}, \ \ \ \dhat^{-1}(u) = \frac{\lambda_{n,p}}{(1/2 - |u|)^2} + \mu_{n,p} 
\end{equation*}

The parameters $\lambda_{n,p},\mu_{n,p},\nu_{n,p}$ depend on the parameters of Binomial distribution $n,p$. Specifically, in H{\"o}rmann's algorithm, the corresponding parameters were chosen to be: $\lambda_{n,p} = - 0.05878 + 0.062744\sqrt{np(1-p)} + 0.01p$, $\mu_{n,p} = 1.15 + 2.53\sqrt{np(1-p)}$, $\nu_{n,p} = np + 0.5$. Importantly, our results are not restricted to any specific choice of hat distribution. Instead, we consider any arbitrary hat distribution $\dhat$ and its inverse $\dinvhat$ that satisfy the conditions of transformed rejection sampling and involve a constant number of basic arithmetic operations.
Therefore, for simplicity and readability, we omit the explicit details of $\dchat$ and $\dhat$ in the rest of the paper. We will refer to $\dinvhat$ as the `inverse function' and $\dhat$ as the `hat function' or `hat distribution'.

The expected runtime of these algorithms is proportional to the rejection rate $\alpha$, and therefore, the runtime is independent of the parameters of the distribution. These algorithms require computing the rejection ratio $r_k = \frac{\bin(k)}{\alpha\dhat(k)}$. But computing this ratio, especially evaluating $\bin(k)$ exactly, can be as expensive as exponential in the number of bits. Therefore, an easily computable approximation $\widetilde{r}_k$ is often obtained to achieve fast scalable practical algorithms. Usually, due to scalability purposes, the logarithm of the rejection ratio $\log r_k$ is computed rather than directly computing $r_k$, which again suffers from other approximation errors due to $\log$ computation. Therefore, these algorithms lack sampling exactly from the distribution.

\subsubsection*{Python implementation of Binomial Sampler}

Standard implementations of Binomial samplers, such as those available in libraries like the GNU Scientific Library (GSL)~\cite{galassi2002gnu}, are designed to work with 64-bit floating-point numbers. Similarly, Python's standard libraries and NumPy~\cite{harris2020array}, provide an implementation of H{\"o}rmann's algorithm for up to the 64-bit floating-point range.  Notably, starting from Python version 3.12, this algorithm has also been integrated into the standard random library of Python, offering support for higher precision computations (though still constrained by the arithmetic computational limits of Python's standard library).

\Cref{alg:btrapprox} presents an abstraction of the standard implementation Binomial Sampler.
Consistent with the current Python implementation, we assume that the \sampler{} algorithm uses the Lanczos approximation to compute the logarithm of the factorial function. To employ the Lanczos approximation, the algorithm uses \logfact{} as described in \cref{alg:logfact}. 
This is one of the most widely used implementations of the Binomial sampling algorithm in practice. We adopt it as a benchmark for developing our error bounds. For clarity, we denote the distribution generated by Python's standard library implementation of \Cref{alg:btrapprox} as $\sampdist$, as opposed to the notation $\bin$, which refers to the Binomial distribution with parameters $n,p$.

\begin{algorithm}
	\caption{$\sampler(n,p)$}\label{alg:btrapprox}
	\SetKwInOut{KwIn}{Input}
	\SetKwInOut{KwOut}{Output}
	\KwIn{Parameters $n,p$}
	\KwOut{Sample $k$}
	\SetAlgoLined
	\DontPrintSemicolon
	Initialize inverse-function-pair $(\dchat, \dhat)$\;
	Initialize rejection parameter $\alpha$\;
	$l_n \gets \logfact(n)$\label{ln:fact1}\;
	\While{True}{
		$v \gets$ uniform random samples within $[0,1]$\;
		$u \gets$ uniform random samples within $[-0.5,0.5]$\;
		$k \gets \lfloor\dinvhat(u)\rfloor$ \label{ln:dinvhat}\;
		$l_k \gets \logfact(k)$, $l_{nk} \gets \logfact(n-k)$\label{ln:fact2}\;    
		$l_v \gets l_n \fminus l_k \fminus l_{nk} \fplus k \ftimes \clog(p) \fplus (n-k) \ftimes \clog(1-p) \fplus \clog(\dhat^{-1}(u)) \fminus \clog(\alpha)$\label{ln:rejratio}\;
		\If{$\clog(v) \leq l_v $}{ \label{ln:comprej}
			\Return $k$
		}    
	}
\end{algorithm}
\subsection{Our Findings}
\label{sec:analysis}
We begin by stating the main theorem of our paper, followed by the supporting lemmas that are used to establish the theorem.

\begin{restatable}{theorem}{main}
	Let the precision of the context be $\beta \geq \max(2\lceil\log_2 n\rceil, 
	\lceil -\log_2 p\rceil)$, and let \(\sampdist\) denote the distribution from 
	which \(\sampler\) samples are drawn. The statistical distance between 
	\(\sampdist\) and \(\bin\) is given by:
	\[
	\dtv\left(\bin, \sampdist\right) \leq (1110 \beta + 3cp + c + \alpha c) n\eps + 15\zeta + \sOh(\eps)
	\] 
	Where $c$ is a constant determined by the inverse function pair $(\dchat, \dhat)$, $\alpha$ is the rejection rate, \(\zeta\) denotes the uniform error bound due to the Lanczos's 
	approximation, \(\eps\) represents the unit round off error $\frac{1}{2^\beta}$, and $\sOh(\eps)$ denotes the higher order terms in $\eps$.
	\label{theo:stirl}
\end{restatable}

The sources of deviations are categorized into two types of errors: 
(1) \erra, which arises from errors in transformed rejection sampling caused by inaccuracies in the computation of $\lfloor\dinvhat(u)\rfloor$ in line~\ref{ln:dinvhat} of \Cref{alg:btrapprox}; and (2) \errb, which refers to errors in the computation of the rejection ratio, accumulating throughout lines~\ref{ln:fact1}, \ref{ln:fact2}, and \ref{ln:rejratio} of \Cref{alg:btrapprox}. 
We will bound the effects of these two errors independently and then combine them to get our main theorem.

\paragraph{Analysis of \erra:} %
The error \erra arises from inaccuracies in the inverse transform sampling, specifically due to deviations in the hat distribution $\dhat$ caused by basic arithmetic operations. The key challenge stems from the inverse sampling procedure evaluating $\dinvhat(u)$ for a uniform random sample $u$. Since evaluating $\dinvhat(u)$ requires basic arithmetic operations, the inverse sampling component may produce an incorrect output $k'$, different from the intended value $k \neq k'$ and thereby deviating $\dhat$ into a modified distribution $\widetilde{\dhat}$.

For $\beta > 2\log_2\lceil n\rceil$, we show that the possible values of 
$k'$ are limited to $\{k-1, k, k+1\}$. Through careful analysis of the 
probabilities of $k'$ being $k-1$ or $k+1$, we show that these probabilities are 
bounded by $\Oh(k\eps)$. This leads to the bound on $\widetilde{\dhat}(k)$ which is stated in the following lemma:

\begin{restatable}{lemma}{dhatdistlem}\label{lem:dhatdist}
    Suppose $\widetilde{\dhat}$ is the deviated version of the hat distribution $\dhat$ due to the error in the computation of $\dinvhat(u)$. Then, if the error is distributed uniformly over the range $[-\epsh, \epsh]$ and $\epsh \leq \frac{1}{2}$, then for any $k \in [n]$,
    $(1-\epsh(3k+1)) \dhat(k) \leq \widetilde{\dhat}(k) \leq (1+w_k\eps_\dchat (k+2)) \dhat(k)$
    where, $w_k = \max\left(\frac{\dhat(k-1)}{\dhat(k)}, \frac{\dhat(k+1)}{\dhat(k)}\right)$
\end{restatable}

\paragraph{Analysis of \errb:}
The error \errb stems from inaccuracies in the rejection ratio computation, which are influenced by factorial approximations, basic arithmetic operations, and \(\log\) computation approximations. Specifically, the rejection ratio error originates from three primary sources: (1) errors introduced by factorial approximations, (2) errors and approximations in basic arithmetic operations and \(\log\) computation during the rejection ratio computation. Both of these errors contribute to the error in the rejection ratio computation. The computed rejection ratio in line~\ref{ln:rejratio} of \Cref{alg:btrapprox} is denoted as $\widetilde{r}_k$, specifically, $\widetilde{r}_k = \frac{\sampdist(k)}{\alpha \dhat(k)}$.

Since Python's standard library uses the Lanczos approximation for \(\log\) computations, we bound the relative error of factorial approximation using the corresponding error bound of the Lanczos approximation (see \cref{sec:factapprox} for details). 
To bound \errb, we also need to understand the relative error in basic arithmetic operations and additive error bounds for logarithm approximations (see \cref{sec:mparithmetic} for details).
By combining these error bounds, we establish a relative error bound for the rejection ratio, formalized in the following lemma.

\begin{restatable}{lemma}{rejratiolem} \label{lem:rejratio}
	Let $r_k$ denote the actual rejection ratio, defined as 
	$r_k = \frac{\bin(k)}{\alpha \dhat(k)}$, and let $\widetilde{r}_k$ 
	denote the computed rejection ratio, defined as 
	$\widetilde{r}_k = \frac{\sampdist(k)}{\alpha \dhat(k)}$, for any $k$ in 
	$[n]$. Then, for any $k \in [0,n]$, we have:
	\[\left|\frac{\widetilde{r}_k}{r_k}-1\right| 
	\leq (1110n + 2540)\beta\eps + 14\eps\log\left( \dhat(k)\right) + 15\zeta + \sOh(\eps)\]
\end{restatable}

By combining the error bounds from both the rejection ratio computation and hat distribution deviation, we obtain the final bound on the statistical distance, as stated in Theorem~\ref{theo:stirl}.

\subsection{Detailed Technical Analysis}
We start by proving the main theorem of our paper using \Cref{lem:dhatdist} and \Cref{lem:rejratio}. 

\begin{proof}[of \cref{theo:stirl}]
    Without loss of generality, assume $\dhat(-1) = \dhat(n+1) = \bin(-1) = \bin(n+1) = 0$ and $r_{-1} = r_{n+1} = 1$.
    Let us define the event $\accept$ as the event such that a sample $k$ is sampled by the sampler. In \sampler{} a sample drawn from $\widetilde{\dhat}$ is accepted with probability $\widetilde{r}_k$.
    Therefore, the acceptance probability is given by, 
    $\Pr(\accept) = \sum_{k=0}^{n} \Pr(\accept | k) \Pr(k)= \sum_{k=0}^{n} \widetilde{r}_k      \widetilde{\dhat}(k)$. Substituing the lower bounds for $\widetilde{r}_k$, $\widetilde{\dhat}(k)$ from \Cref{lem:rejratio} and \Cref{lem:dhatdist} we get: 
    \begin{align*}
        \Pr(\accept) 
        &\geq \sum_{k=0}^{n} (1 - (1110n+2540)\beta\eps- 14\eps\log\left( \dhat^{-1}(u)\right) -15 \zeta -\sOh(\eps)) \\
        &\hspace{18em}(1 - (3k+1)\epsh) r_k h(k) \\
        &\geq \sum_{k=0}^{n} (1 - (1110n+2540)\beta\eps - 14\eps\log\left( \dhat^{-1}(u)\right) -15\zeta -\sOh(\eps)) \\
        &\hspace{18em}(1 - (3k+1)\epsh) \frac{\bin(k)}{\alpha} 
    \end{align*}
    To complete the lower bound we first observe that,
    \[
    \sum_{k=0}^{n} \log(\dhat^{-1}(u)) \bin(k)
    = \sum_{k=0}^{n} \log\left(\frac{\bin(k)}{\dhat(k)}\right) \bin(k) 
    - \sum_{k=0}^{n} \log(\bin(k)) \bin(k).
    \]
    Given that $\frac{\bin(k)}{\dhat(k)} \leq \alpha$ for all $k$, and that the entropy of Binomial distribution satisfies
    $\E[-\log(\bin(k))] \leq \frac{1}{2} \log_2(2\pi e np(1-p))$,
    we conclude that
    \[
    \sum_{k=0}^{n} \log(\dhat^{-1}(u)) \bin(k)
    \leq \alpha + \frac{1}{2} \log_2(2\pi e np(1-p)).
    \]
    Thus we can lower bound acceptance probability $\Pr(\accept)$ as follows,
    \begin{align*}
    \Pr(\accept)
    &\geq \left(1 - (1110n + 2554)\beta\eps - 14\alpha\eps + \log\left(\dhat^{-1}(u)\right) 
    - 15\zeta \right. \\
    &\hspace{15em} \left. - (3np + 1)\epsh - \sOh(\eps)\right) \cdot \frac{1}{\alpha}
    \end{align*}
    Therefore, the probability of observing a point $k$ under the sampler’s output distribution is given by $\sampdist(k) = \Pr(k | \accept)$. Applying Bayes' rule we get, $\sampdist(k) = \frac{\Pr(\accept | k) \Pr(k)}{\Pr(\accept)} = \frac{\widetilde{r_k} \widetilde{\dhat}(k)}{\Pr(\accept)}$. By using the upper bound of $\widetilde{r}_k$, from \Cref{lem:rejratio},
    \begin{align*}
        \sampdist(k) &\leq \frac{1 + (1110n+2540)\beta\eps +14\eps\log\left( \dhat^{-1}(u)\right) + 15\zeta +\sOh(\eps)}{1 - (1110n+2554)\beta\eps -15\zeta -(3np+1)\epsh - \sOh(\eps)} 
        \cdot \alpha r_k \widetilde{\dhat}(k) \\
        &\leq (1 + (2220n+5080)\beta\eps +28\eps\log\left( \dhat^{-1}(u)\right) + 30\zeta \\
        &\hspace{13em} + (6np+2)\epsh +\sOh(\eps)) \cdot \frac{\bin(k)\widetilde{\dhat}(k)}{\dhat(k)}  
    \end{align*}
    The last inequality follows from assuming $(1110n+2540)\beta\eps +(3np+1)\epsh +14\eps\log\left( \dhat^{-1}(u)\right) + 15\zeta +\sOh(\eps) \leq \frac{1}{2}$.
    Next we bound the ratio $\frac{\widetilde{\dhat}(k)}{\dhat(k)}$ using \Cref{lem:dhatdist},
    \begin{align*}
        \sampdist(k) &\leq \Bigl(1+ (2220n+5080)\beta\eps +28\eps\log\left( \dhat^{-1}(u)\right)+ 30\zeta \\
        &\hspace{6em}+ (6np+2)\epsh +\sOh(\eps)\Bigr) \cdot (1 + w_k\eps_\dchat (k+2)) \cdot \bin(k)\\
        &\leq \Bigl(1+ (2220n+5080)\beta\eps +28\eps\log\left( \dhat^{-1}(u)\right)+ 30\zeta + (6np+2)\epsh +  \\
        &\qquad\qquad\qquad\qquad\qquad\qquad\qquad\qquad w_k\eps_\dchat (k+2) + \sOh(\eps)\Bigr) \cdot \bin(k)
    \end{align*}
    Where $w_k = \max\left(\frac{\dhat(k-1)}{\dhat(k)}, \frac{\dhat(k+1)}{\dhat(k)}\right)$.
    Since $r_k \leq 1$, it follows that $\dhat(k) \geq \frac{\bin(k)}{\alpha}$ and $\dhat(k+1) = \frac{\bin(k+1)}{\alpha r_{k+1}}$. This implies $\frac{\dhat(k+1)}{\dhat(k)} \leq \frac{\bin(k+1)}{\bin(k) r_{k+1}}$. Similarly, we have $\frac{\dhat(k-1)}{\dhat(k)} \leq \frac{\bin(k-1)}{\bin(k) r_{k-1}}$. Combining these, $w_k$ can be upper bounded as follows, 
    \begin{align*}
        w_k &= \max\left(\frac{\dhat(k-1)}{\dhat(k)}, \frac{\dhat(k+1)}{\dhat(k)}\right) \leq \max\left(\frac{\bin(k-1)}{\bin(k)r_{k-1}}, \frac{\bin(k+1)}{\bin(k)r_{k+1}}\right)    
    \end{align*}
    This yields the following bound: $w_k \bin(k) \leq \max\left(\frac{\bin(k-1)}{r_{k-1}}, \frac{\bin(k+1)}{r_{k+1}}\right)$.
    Thus, $w_k \bin(k) \leq \max\left(\dhat(k-1), \dhat(k+1)\right) \leq \dhat(k-1)+ \dhat(k+1)$.
    Using this bound, we can upper bound the sum $\sum_{k=0}^{n} w_k\eps_\dchat (k+2) \bin(k)$ as follows, 
    \begin{align*}
        \sum_{k=0}^{n} w_k\eps_\dchat (k+2) \bin(k) &\leq \alpha\epsh \sum_{k=0}^{n} (k+2) \left(\dhat(k-1)+ \dhat(k+1)\right)\\
        &\leq 2\alpha\epsh(n+2)
    \end{align*}
    The last inequality follows from the fact that $k \leq n$, $\sum_{i=0}^{n-1} \dhat(i) \leq 1$ and $\sum_{i=1}^{n} \dhat(i) \leq 1$.  
    Therefore, the statistical distance between the sampler's distribution and the Binomial distribution is given by,
    \begin{align*}
        &\dtv(\sampdist, \bin) = \frac{1}{2} \sum_{k=0}^{n} \left|\sampdist(k) - \bin(k)\right|\\
        \leq& \frac{1}{2} \sum_{k=0}^{n} \left|(2220n+5080)\beta\eps +28\eps\log\left( \dhat^{-1}(u)\right) + 30\zeta + (6np+2)\epsh \right.\\
        &\hspace{15em} \left. + 2\alpha\epsh(n+2) + \sOh(\eps) \right| \cdot \bin(k) \\
        \leq& (1110n+2554)\beta\eps + 14\alpha\eps + 15\zeta + (3np + 1)\epsh + \alpha\epsh(n+2) + \sOh(\eps)\\
        =& (1110 \beta + 3cp + c + \alpha c) n\eps + 15\zeta + \sOh(\eps)
    \end{align*} 
    This completes the proof. \qed
\end{proof}

\subsubsection{Error in Inverse Transform Sampling}

In this subsection, we analyze the error introduced during the inverse transform sampling process which arises from the arithmetic operations involved in evaluating the inverse hat function $\dinvhat(u)$. Instead of assuming a specific hat distribution, we argue that for any hat distribution, the multiplicative error introduced during the computation of $\dinvhat(u)$ is multiplicatively bounded by $\epsh:= c\eps$, where $c>0$ depends on the number of basic arithmetic operations\footnote{For most practical hat functions, we typically have $c \leq 100$, which implies that $\epsh \leq \frac{1}{2}$ when the precision parameter $\beta > 8$.} 
$\ast$ used in $\dinvhat$ with $\ast \in \{+, -, \times, /, \sqrt{}\}$. 
Consequently, the computed value of $\dinvhat(u)$ in \cref{ln:dinvhat} of \sampler{} falls within the range $[(1 - \epsh) \dinvhat(u), (1 + \epsh) \dinvhat(u)]$. 
\crdy{To model the impact of the error we assume it to be uniformly distributed over the range $[-\epsh, \epsh]$. 
While a Gaussian distribution might seem like a natural choice for such errors, the uniform distribution offers a more conservative approach. 
Among all zero-mean Gaussian like distributions with bounded support, the uniform distribution has the heaviest tails within its support. This characteristic makes it suitable for handling errors in the worst-case scenario.
In particular, if the mean of the error is centered at $\dinvhat(u)$, then assuming a uniform distribution for the error allows us to upper bound deviations in $\dhat(k)$ more conservatively. This ensures that our bounds remain valid even under pessimistic error assumptions.}
We restate the Lemma~\ref{lem:dhatdist} for completness and provide its proof in the full version.

{\dhatdistlem*}

\subsubsection{Error in Rejection Ratio Computation}
In this subsection, we will analyze the impact of factorial approximations, basic arithmetic operations, and $\log$ computation approximations on the rejection sampling process. Before proceeding further, we introduce the notation $\interdist$, which we refer to it as the \textit{intermediate} distribution, where only the factorial computations are approximated. The key result of this subsection is \Cref{lem:rejratio}, whose proof builds on two auxiliary lemmas: one addressing errors from arithmetic and $\log$ approximations, and another handling factorial approximation. Proofs are deferred to the full version due to space constraints.

\begin{restatable}{lemma}{rejratiointer}\label{lem:rejratiointer}
	Let $\rinter$ denotes the ratio 
	$\rinter = \frac{\interdist(k)}{\alpha \dhat(k)}$, and let $\widetilde{r}_k$ 
	denote the computed rejection ratio, defined as 
	$\widetilde{r}_k = \frac{\sampdist(k)}{\alpha \dhat(k)}$, for any $k$ in 
	$[n]$. Then, for all $k \in [n]$,
	$\left| \frac{\widetilde{r}_k}{\rinter} - 1\right| 
	\leq (1110n + 2540)\beta\eps + 14\eps\log\left( \dhat(k)\right) + \sOh(\eps)$  
\end{restatable}

\begin{restatable}{lemma}{sampdistlem}\label{lem:lancz}
    For all $k \in [n]$, 
    $
    (1-15\zeta) \bin(k) \leq \interdist(k) \leq (1+15\zeta) \bin(k)
    $,
    where \(\zeta\) denotes the uniform error bound due to the Lanczos approximation and $\interdist$ denotes the intermediate distribution. 
\end{restatable}

For completeness, we restate the \Cref{lem:rejratio} whose proof follows by combining \Cref{lem:rejratiointer} and \Cref{lem:lancz}.

\rejratiolem*
\section{Case Study: DNF Counting}\label{sec:casestudy}

This case study demonstrates how our proposed bounds on the statistical distance between the sampler and the Binomial distribution can be easily integrated into practical tools. To demonstrate the applicability of these bounds, we utilize them in conjunction with an off-the-shelf Binomial sampler to implement the DNF Counting algorithm \aps{}~\cite{meel2021estimating}.

A DNF formula is a disjunction of conjunctions of literals, where each conjunction (clause) represents a set of conditions. For example, \((x_1 \land \neg x_2) \lor (x_2 \land \neg x_3)\) is a DNF formula with clauses like \((x_1 \land \neg x_2)\). A DNF formula \(\varphi := \varphi_1 \lor \ldots \lor \varphi_m\) has \(m\) clauses, and the number of its solutions is denoted as \(|sol(\varphi)|\).
The problem of counting the $|sol(\varphi)|$ for a DNF formula $\varphi$ is \#P-hard. To address this challenge, various Fully Polynomial-Time Randomized Approximation Schemes (FPRAS) have been developed~\cite{karp1983monte,karp1989monte,chakraborty2016algorithmic}. Given a DNF formula $\varphi$ and tolerance and confidence parameters $\eps, \delta\in [0,1]$, these FPRAS return $\widehat{n} \in \left[(1-\eps)|sol(\varphi)|, (1+\eps)|sol(\varphi)|\right]$ with probability at least $(1-\delta)$. 

The most recent progress in sampling-based DNF counting FPRAS is embodied by the algorithm \aps{}~\cite{meel2021estimating}. Given a DNF formula $\varphi$, the \aps{} returns an $\eps$ multiplicative approximation of $|sol(\varphi)|$ with high probability.
The algorithm maintains a bucket $\cX$ to keep sampled solutions from DNF clauses and, also, a probability parameter $p$ such that, for any solution $\sigma$ of $\varphi$, $\sigma$ belongs to $\cX$ with probability $p$. To achieve this goal, the algorithm removes all the elements $\sigma$ from bucket $\cX$ if $\sigma$ satisfies $\varphi_i$. The algorithm next samples new solutions from $\varphi_i$. To determine the number of solutions, the \aps{} asks for a sample $N_i$ from Binomial distribution $\mathsf{b}_{|sol(\varphi_i)|, p}$ and adds $N_i$ many new satisfying assignments of $\varphi_i$ to $\cX$. If the bucket overflows, the algorithm keeps on removing elements uniformly from the bucket until the bucket size falls under the threshold. The end goal of \aps{} is to output the ratio $\frac{|\cX|}{p}$ which is a good estimate of $|sol(\varphi)|$.

Since \aps{} heavily relies on the Binomial sampler, the theoretical guarantees of \aps{} are contingent on the quality of the Binomial sampler. 
This case study illustrates how our results allow users to maintain the theoretical guarantees of {\aps}.
Computing bounds on $\dtv$ between the Binomial distribution and the sampler, users can adjust the confidence parameter $\delta$ in \aps{} to account for errors from the underlying Binomial sampler, thereby ensuring correctness with theoretical guarantees.

\begin{algorithm}[htb]
	\SetAlgoLined
	\DontPrintSemicolon

	\HiLi $\delta_1 \gets \kappa\delta$\;
	\HiLi  $\delta_2 \gets (1-\kappa)\delta$

	Initialize $T \gets \left(\frac{\log(4/\delta_2) + \log m}{\eps^2}\right)$\; Initialize $p\gets 1$, $\cX \gets \emptyset$\;
	$\delta' \gets 0$\;
	\For{$i=1$ to $m$}{
		\For{all $\sigma \in \mathcal{X}$}{
			\If{$\sigma \vDash \varphi_i$} {
				remove $\sigma$ from $\mathcal{X}$\;
			}
		}
		$N_i \gets \sampler(|sol(\varphi_i)|, p)$ \label{ln:samp}\;
		\HiLi $\delta' \gets \delta' + \delta^i_{|sol(\varphi_i)|, p}$\;
		\HiLi   \If{$\delta' > \delta_1$}{
			\HiLi         \Return \fail{}
		}
		Add $N_i$ distinct random solutions of $\phi_i$ to $\cX$\;
		\While{$|\cX| > T$}{
			$p=p / 2$\;
			Throw away each element of $\mathcal{X}$ with probability $\frac{1}{2}$\;
		}
	}
	Output $\frac{|\cX|}{p}$\;
	\caption{\apsnew$(\varphi, \eps, \delta, \kappa)$ \newline \small{(The modifications from \aps{} to \apsnew{} are highlighted)}}
	\label{alg:apsnew}
\end{algorithm}

\noindent \textbf{ Algorithm Modification}
We demonstrate how easily the \aps{} algorithm can be modified to incorporate our statistical distance bounds.
We denote this modified version as \apsnew{}, detailed in \cref*{alg:apsnew}, with highlighted modifications. The primary difference between \aps{} and \apsnew{} lies in handling the confidence parameter \( \delta \) to account for the errors due to the underlying binomial sampler.  Therefore, \apsnew{} takes another parameter \( \kappa \in [0,1] \)  to adjust the error budget for the sampler, such that, $\kappa \delta$ is the error budget for the sampler and $(1-\kappa)\delta$ is the error budget for the algorithm. If the accumulated error due to the sampler exceeds the error budget, \apsnew{} halts immediately and returns \fail{}. The user can restart the algorithm using a larger value of \( \kappa \) to accommodate an increased error budget.

\noindent \textbf{Analysis Modification}
We now illustrate how simply the analysis of \aps{} can be modified to \apsnew{}. 
Recall that we are concerned with the event: $\frac{|\mathcal{X}|}{p} \notin (1 \pm \eps)|sol(\varphi)|$, which we will refer to as $\mathsf{Bad}$. 
Note that, $\delta_1 = \kappa \delta$ is the error budget for the sampler and $\delta_2 = (1-\kappa)\delta$ is the error budget for the algorithm. 
By the correctness guarantee of \aps{}, $\Pr_{\bin}(\mathsf{Bad}) \leq \delta_2$.
During the execution of \apsnew{}, if the algorithm invokes the sampler $t$ times with parameters $n_i,p_i$, then from \Cref{lem:indistinguishable}, we have 
$$\Pr_{\sampdist}(\bad) \leq \Pr_\bin(\bad) + \sum_{i = 1}^t\dtv(\bb_{n_i, p_i}, \bb^{\mathsf{BinSamp}}_{n_i, p_i})$$
Suppose during the execution, the computed $\dtv$ bounds using \Cref{theo:stirl} are given by $\delta^1_{n_1, p_1}, \delta^2_{n_2, p_2}, \ldots, \delta^t_{n_t, p_t}$ such that $\delta' = \sum_{i = 1}^{t}\delta^i_{n_i, p_i}$. Therefore, if \apsnew{} does not halt then $\Pr_{\sampdist}(\bad) \leq \Pr_\bin(\bad) + \sum_{i = 1}^t\dtv(\bin, \sampdist) \leq \delta_2 + \delta' \leq \delta_2 + \delta_1 \leq \delta$. Thus, by the correctness of \aps{} the count returned by \apsnew{} is within the $\eps$ error bound. If \apsnew{} halts and returns \fail{}, this implies that the error budget of the sampler has been exceeded.

\def\HiLi{\leavevmode\rlap{\hbox to \hsize{\color{gray!20}\leaders\hrule height .8\baselineskip depth .5ex\hfill}}}

\subsection{Experimental Setup and Evaluation Results}
We conducted accuracy experiments following the methodology of \cite{soosengineering} to assess the reliability of the counts returned by \apsnew{}. Specifically, we compared the counts from \apsnew{} with those from \ganak{}~\cite{sharma2019ganak}, an exact counting tool.
We used a comprehensive benchmark suite from \cite{soosengineering} to evaluate the performance and accuracy of the algorithms. \crdy{This suite consists of DNF formulas with the number of variables ranging from 100 to 700 and clause counts ranging from 30 to 700. Following prior work \cite{soosengineering,meel2019not,soos2019bird}, we adopted the standard settings for the tolerance parameter ($\eps$ = 0.8) and the confidence parameter (\(\delta\) = 0.36), which are commonly used in both model counting competitions and practical applications.} Finally, we set \(\kappa\) to 0.5.

\begin{figure}[tb]
	\centering
	\includegraphics[scale=0.4]{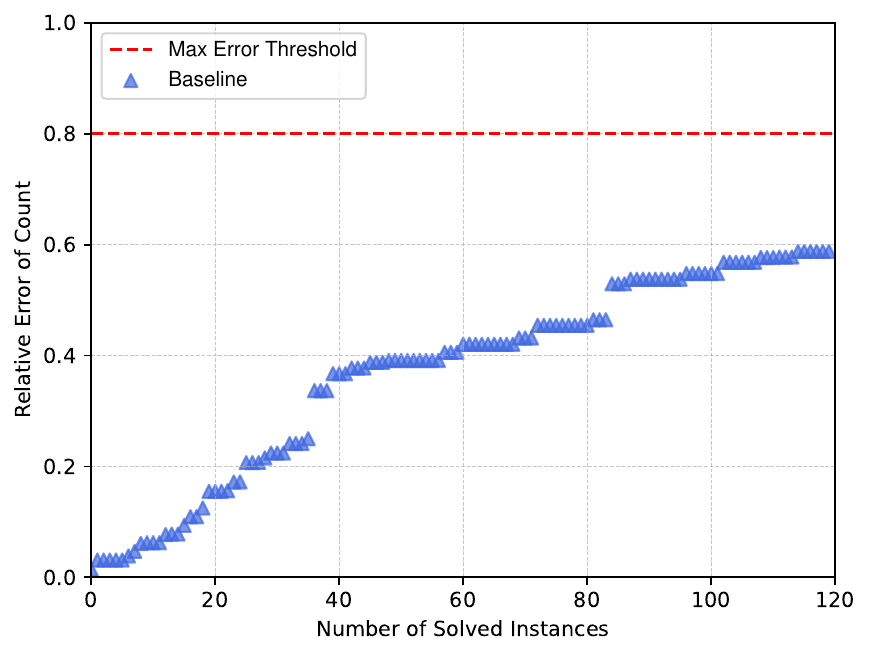}
	\caption{Accuracy experiment results for \apsnew{}. The red line represents the tolerance factor (\(\eps = 0.8\)).}
	\label{fig:apsresultacc}
\end{figure}

\paragraph{Summary of Results}
Our observations reveal that \apsnew{} delivers counts that are nearly as precise as the exact counts obtained from \ganak{}, as demonstrated in \cref{fig:apsresultacc}. 
The \(y\)-axis of \cref{fig:apsresultacc} represents the relative error of the counts, with the tolerance parameter (\(\eps = 0.8\)) marked by a red straight line, while the \(x\)-axis represents the instances. We observed that for all instances, \apsnew{} computed counts within the tolerance, demonstrating high accuracy. 

\crdy{\Cref{fig:apsresulterr} presents the reported errors $\delta'$ of the Binomial sampler during the execution of \apsnew{}. The left plot presents the reported errors for individual instances, while the right plot groups the average error by the number of clauses in the DNF formula.  Across our benchmark suite, the errors remain around \(10^{-6}\). The error increases with the number of clauses as the number of calls to the Binomial sampler increases.  Notably, we observe a 10-fold increase in error when the number of clauses grows from 30 to 700.}

The results suggest that \apsnew{} is capable of providing highly reliable approximations of counts across a diverse range of DNF instances, maintaining accuracy within the predefined error margin despite its simple design. The integration of our statistical distance bounds into the algorithm allows users to effectively manage the error budget, ensuring that the algorithm remains robust and reliable even in the presence of approximation errors from the underlying Binomial sampler.

\begin{figure}[tb]
	\centering
	\includegraphics[scale=0.4]{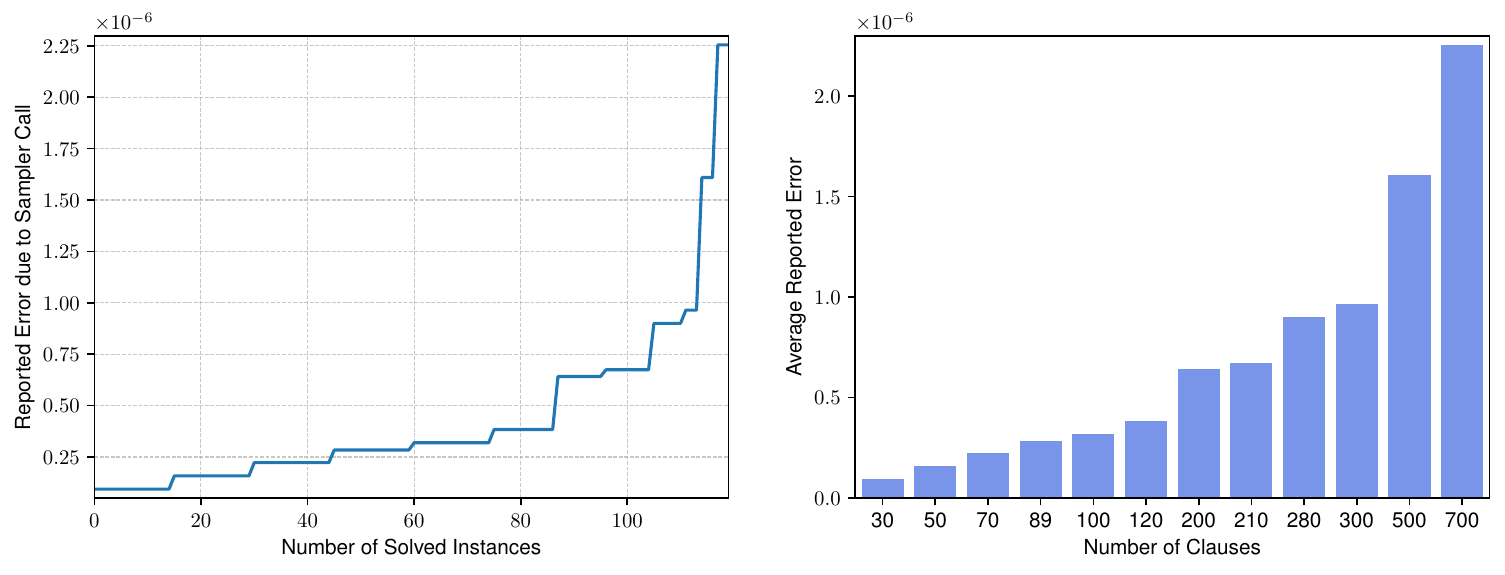}
	\caption{Left: Reported errors $(\delta')$ by \apsnew{} for individual instances. Right: Average error grouped by the number of clauses in the DNF formula, showing a clear upward trend as formula size increases.}
	\label{fig:apsresulterr}
\end{figure}

\section{Conclusion}
\label{sec:conclusion}
In this work, we first identified the sources of deviation in the practical implementations of standard binomial samplers. We observed that exact sampling from distributions is infeasible in practice due to high runtime overhead; thus, implementations inevitably introduce deviations. Accordingly, we proposed the usage of statistical distance as the quality metric owing to its ability to allow end users to obtain sound bounds on the {\em bad} events. We also presented a case study demonstrating the minimal effort required by system designers to incorporate the reported deviation bounds into their systems. 

\paragraph*{Limitations and Future Work} While our current work establishes a foundational framework, there are several limitations and opportunities for future enhancement.
The current analysis relies on several simplifying assumptions—for instance, uniformity in the error distribution—which may not hold in more general settings. Additionally, the reported bounds are not yet tight and can be refined for greater accuracy, making this an important avenue for further research. Moreover, the principles of quality measurement can be extended to samplers for other distributions. Developing a general framework for error reporting across various types of samplers would be a valuable contribution to the field. Finally, proposing efficient sampling scheme that can achieve the desired statistical distance with minimal overhead is an open problem that warrants further investigation.

\subsubsection*{\ackname}
Part of this research is supported by DCSW-TAC Project ACMU-24/VPP (SC). 
Meel acknowledges the support of the Natural Sciences and Engineering Research Council of Canada (NSERC), funding reference number RGPIN-2024-05956. 
Sarkar acknowledges the support of Google PhD Fellowship. 
Part of the research was conducted while Sarkar was at the University of Toronto. 
Computations were performed on the Niagara supercomputer at the SciNet HPC Consortium. SciNet is funded by Innovation, Science and Economic Development Canada; the Digital Research Alliance of Canada; the Ontario Research Fund: Research Excellence; and the University of Toronto. 

\subsubsection*{\discintname}
The authors declare that they have no competing interests.

\newpage

\appendix

\newpage

\section*{Appendix}

\section{Extended Proofs}

\subsection*{Proof of \Cref{lem:rejratiointer}}
To prove \Cref{lem:rejratiointer}, we begin by analyzing how $\log n!$ is approximated. In this context, we refer to \logfact{} (\Cref{alg:logfact}), which presents an approximation of $\log n!$.
The following lemma provides a bound on the error introduced by using $\clog$ within one \logfact{} call. 

\begin{lemma} \label{lem:logfact}
    The additive error introduced by using $\clog$ in a single \logfact{} call can be bounded by $(k+2)\tau$. Specifically,
    \begin{align*}
    	 \log(\lancz(k)) - (k+2)\tau \leq \logfact(k) \leq \log(\lancz(k)) + (k+2)\tau
    \end{align*}
\end{lemma}
\begin{proof}
	Given $k$, $\logfact(k)$ computes and returns
     $\frac{1}{2} \clog{}{}(2\pi) + \left(k + \frac{1}{2}\right)\clog{}{} (k + g + \frac{1}{2}) - \left(k + g + \frac{1}{2}\right) + \clog \left(A_{t,g}(k)\right)$. 
    Therefore, from \cref*{eq:approxlog} we have the following inequality,
    \begin{align*}
        &\frac{1}{2} \clog{}{}(2\pi) + \left(k + \frac{1}{2}\right)\clog{}{} \left(k + g + \frac{1}{2}\right) - \left(k + g + \frac{1}{2}\right) + \clog \left(A_{t,g}(k)\right)\\
        &\leq \frac{1}{2} \log{}{}(2\pi) + \frac{1}{2}\tau + \left(k + \frac{1}{2}\right)\log{}{} \left(k + g + \frac{1}{2}\right) + \left(k + \frac{1}{2}\right)\tau \\
        & \qquad - \left(k + g + \frac{1}{2}\right) + \log \left(A_{t,g}(k)\right) + \tau\\
        &\leq \log (\lancz(k!)) + (k+2)\tau 
    \end{align*}
    Similarly, $\frac{1}{2} \clog{}{}(2\pi) + \left(k + \frac{1}{2}\right)\clog{}{} (k + g + \frac{1}{2}) - \left(k + g + \frac{1}{2}\right) + \clog \left(A_{t,g}(k)\right) \geq \log (\lancz(k!)) - (k+2)\tau$. \qed
\end{proof}

\rejratiointer*
\begin{proof}
    \sampler{} uses the hat distribution $\dhat$.
    We use $\widetilde{lr}_k$ to denote the log of the computed rejection ratio $\widetilde{r}_k$, which is given by,
	\begin{equation*}\label{eq:clogfact}
        \widetilde{lr}_k := l_n \fminus l_k \fminus l_{nk} \fplus k \ftimes \clog(p) \fplus (n-k) \ftimes \clog(1-p) \fplus \clog(\dhat^{-1}(u))
    \end{equation*}
    We assume that the computation of $\dhat^{-1}(u)$ has constant number $c>0$ of arithmetic operations. Therefore the multiplicative error introduced in the computation of $\dhat^{-1}(u)$ is bounded by $c\eps$. Therefore, the computed value of $\clog \left(\dhat^{-1}(u)\right)$ is bounded by $\clog \left((1\pm c\eps) . \dhat^{-1}(u)\right)$. Consequently, using \cref{eq:approxlog},
    \begin{equation}
        \label{eq:log:hu}
        \clog \left((1+ c\eps) . \dhat^{-1}(u)\right) \leq \log\left(\dhat^{-1}(u)\right) + \log(1+ c\eps) + \tau    
    \end{equation}
    Similarly, considering the error in the term $k\ftimes\clog{}(p)$, we deduce using \cref{eq:basicops}, \cref{eq:approxlog} that,
    \begin{equation}
        \label{eq:log:nkpq}
        k \ftimes \clog{}(p) \leq k \log(p) + \eps |k \log(p)| + k \tau + k \tau\eps
    \end{equation}
    Similarly, we can bound the term $(n-k) \ftimes \clog{}(1-p)$ as follows,
    \begin{equation}
        \label{eq:log:nk1p}
        (n-k) \ftimes \clog{}(1-p) \leq (n-k) \log(1-p) + \eps |(n-k) \log(1-p)| + (n-k) \tau + (n-k) \tau\eps
    \end{equation}
    Lastly, we are interested in the error in the computation of $l_n - l_k - l_{nk}$. Using \cref{lem:logfact}, we can bound the error as follows:
    \begin{align}
        \label{eq:log:lklknk}
        l_n - l_k - l_{nk} &\leq \log \lancz(n) - \log \lancz(k) - \log \lancz(n-k) \notag \\
        &\qquad + (n+2)\tau + (k+2)\tau + (n-k+2)\tau
    \end{align}
    Thus, we bound the following sum using \cref{eq:log:nkpq}, \cref{eq:log:nk1p} and \cref{eq:log:lklknk}:
    \begin{align}
        \label{eq:terribl1}
        &l_n - l_k - l_{nk} + k \ftimes \clog(p) + (n-k) \ftimes \clog(1-p) \notag\\
    \leq& \log \lancz(n) - \log \lancz(k) - \log \lancz(n-k)  + k \ftimes \clog(p) \notag\\
        &\qquad + (n-k) \ftimes \clog(1-p) + (n+2)\tau + (k+2)\tau + (n-k+2)\tau \notag\\
        =& \log \lancz(n) - \log \lancz(k) - \log \lancz(n-k)  + k\log p + \eps|k\log p| \notag\\
        &\qquad + (n-k)\log (1-p) + \eps|(n-k)\log (1-p)|  +  (3n + 6)\tau + n\tau\eps
    \end{align}
    In a similar way $l_n + l_k + l_{nk} + k \ftimes \clog(p) + (n-k) \ftimes \clog(1-p)$ is at most,
    \begin{align}
        \label{eq:terribl2}
        & \log \lancz(n) + \log \lancz(k) + \log \lancz(n-k) + k \ftimes \clog(p) \notag\\
        &\qquad + (n-k) \ftimes \clog(1-p) + (n+2)\tau + (k+2)\tau + (n-k+2)\tau \notag\\
    \leq& n\log n + k\log k + (n-k)\log(n-k) + k\log p + \eps|k\log p| \notag\\
        &\qquad + + (n-k)\log (1-p) + \eps|(n-k)\log (1-p)| + (3n +6)\tau + n\tau\eps
    \end{align}
    The last inequality follows from the fact that $\log\lancz(n) \leq n\log n$. 
    Next we accumulate all the errors due to basic arithmetic computations in the computation of $\widetilde{lr}_k$. We first bound the compound sum using \cref{eq:basicops-sum}:
    \begin{align*}
        & l_n \fminus l_k \fminus l_{nk} \fplus k \ftimes \clog(p) \fplus (n-k) \ftimes \clog(1-p) \fplus \clog\left( \dhat^{-1}(u)\right)\\
        \leq & l_n - l_k - l_{nk} + k \ftimes \clog(p) + (n-k) \ftimes \clog(1-p) + \clog\left( \dhat^{-1}(u)\right)\\
        &\qquad + 6\eps \left(l_n + l_k + l_{nk} + k \ftimes \clog(p) + (n-k) \ftimes \clog(1-p) + \clog\left( \dhat^{-1}(u)\right)\right)
    \end{align*}
    Using the corresponding bounds from \cref{eq:terribl1,eq:terribl2} we can have,
    \begin{align*}
        & l_n \fminus l_k \fminus l_{nk} \fplus k \ftimes \clog(p) \fplus (n-k) \ftimes \clog(1-p) \fplus \clog\left( \dhat^{-1}(u)\right)\\
        \leq & \log \lancz(n) - \log \lancz(k) - \log \lancz(n-k) \\
        &\qquad + k \log(p) + (n-k) \log(1-p) + \log\left( \dhat^{-1}(u)\right)\\
        &\qquad + 7\eps \left[n\log n + k\log k + (n-k)\log(n-k) + k |\log(p)| \right.\\
        &\qquad \left. + (n-k) |\log(1-p)| + \log\left( \dhat^{-1}(u)\right)\right] + (3n+7)\tau + \log(1+c\eps)  + \Oh(n\tau\eps)
    \end{align*}
    Thus accumulating all the errors, and denoting $\log(\rinter)$ by $lr_k$, we can bound the error in the computation of $\widetilde{lr}_k$ using the following inequalities:
    \begin{align*}
        \widetilde{lr}_k &\leq lr_k + (3n + 7)\tau + 7\eps \left[n\log n + k\log k + (n-k)\log(n-k) \right.\\
        & \left. + k |\log(p)| + (n-k) |\log(1-p)| + \log\left( \dhat^{-1}(u)\right)\right] + \log (1+c\eps) + \Oh(n\tau\eps)\\
        &\leq lr_k + (3n + 7)\tau + 21n\beta\eps + 7\eps \log\left( \dhat^{-1}(u)\right) + \log (1+c\eps) + \Oh(n\tau\eps)
    \end{align*}
    The second inequality follows from $\log n, \log k, \log(n-k), |\log p|, |\log (1-p)|< \beta$ (due to our choice of $\beta$).
    Using a similar set of inequalities, we can bound the error in the computation of $\widetilde{lr}_k$ from below.
    \begin{align*}
        \widetilde{lr}_k &\geq lr_k - (3n +7)\tau - 21n\beta\eps -7\eps\log\left( \dhat^{-1}(u)\right) - \log (1-c\eps) - \Oh(n\tau\eps)
    \end{align*}
    Therefore, considering $\widetilde{r}_k = e^{\widetilde{lr}_k}$, and using inequalities $e^x \leq  1+2x, e^{-x} \geq  1-2x$, we get the following bound on the obtained rejection ratio:
    \begin{align*}
        \widetilde{r}_k &\leq r_k \cdot (1+c\eps) \cdot e^{(3n + 7)\tau + 21n\beta\eps +7\eps\log\left( \dhat^{-1}(u)\right)+ \Oh(n\tau\eps)}\\
        &\leq r_k (1+(6n + 14)\tau + 42n\beta\eps + 14\eps\log\left( \dhat^{-1}(u)\right) + c\eps + \Oh(n\tau\eps))\\
        &\leq r_k (1+178(6n + 14)\beta\eps + 42n\beta\eps + 14\eps\log\left( \dhat^{-1}(u)\right)+ c\eps + \Oh(n\eps^2))\\
        &= r_k (1+ (1110n + 2540)\beta\eps + 14\eps\log\left( \dhat^{-1}(u)\right) + \Oh(n\eps^2))
    \end{align*}
    The last inequality follows from substituing the value of $\tau$.
    Similarly, we have, $\widetilde{r}_k \geq r_k (1 - (1110n + 2540)\beta\eps - 14\eps\log\left( \dhat^{-1}(u)\right) - \Oh(n\eps^2))$. \qed
\end{proof}

\subsection*{Proof of \Cref{lem:lancz}}

\sampdistlem*
\begin{proof}
    We begin by denoting the approximated probability mass, as calculated by \sampler{}, with $\barbin(k)$ for all $k \in [n]$, such that, 
    $$\barbin(k) = \frac{\lancz(n)}{\lancz(k) \cdot \lancz(n-k)} \cdot p^k q^{n-k}$$
    But, here $\barbin$ is not necessarily a well defined distribution, since $\sum_{k=0}^n \barbin{k}$ is not necessarily $1$. We normalize $\barbin$ to a distribution and obtain its upper and lower bounds. From \cref{eq:lanczerr} we have $(1-\zeta) k! \leq \lancz(k) \leq (1+\zeta) k!$ for all $k \in [n]$. Therefore,
    $$
    \begin{aligned}
        \sum_{k = 0}^n \bin(k) \cdot \frac{1-\zeta}{(1+\zeta)^2} \leq \sum_{k = 0}^n \barbin(k) &\leq \sum_{k = 0}^n \bin(k) \cdot \frac{1+\zeta}{(1-\zeta)^2} 
    \end{aligned}
    $$
    Therefore for all $k \in [n]$ we have the following sets of inequalities:
    $$
    \begin{aligned}
    \frac{(1-\zeta)}{(1+\zeta)^2} \bin(k) \cdot \frac{(1-\zeta)^2}{(1+\zeta)} &\leq \frac{\barbin(k)}{\sum_{i = 0}^n \barbin(i)} &&\leq \frac{(1+\zeta)}{(1-\zeta)^2} \bin(k) \cdot \frac{(1+\zeta)^2}{(1-\zeta)}\\
    \implies \qquad \frac{(1-\zeta)^3}{(1+\zeta)^3} \bin(k) &\leq \frac{\barbin(k)}{\sum_{i = 0}^n \barbin(i)} &&\leq \frac{(1+\zeta)^3}{(1-\zeta)^3} \bin(k)\\
    \implies \qquad (1-3\zeta)^3 \bin(k) &\leq \frac{\barbin(k)}{\sum_{i = 0}^n \barbin(i)} &&\leq (1+3\zeta)^3 \bin(k) \\
    \implies \qquad (1-15\zeta) \bin(k) &\leq \frac{\barbin(k)}{\sum_{i = 0}^n \barbin(i)} &&\leq (1+15\zeta) \bin(k) 
    \end{aligned}
    $$

    The third and the last inequalities follow due to the fact that for $\zeta < 1/3$, $\frac{1-\zeta}{1+\zeta} > 1-3\zeta$, $(1-3\zeta)^3 > 1-15\zeta$ and $(1+3\zeta)^3 < 1+15\zeta$.

    Since, $\interdist(k) = \frac{\barbin(k)}{\sum_{i = 0}^n \barbin(i)}$, the result follows directly. \qed
\end{proof}

\subsection*{Proof of \Cref{lem:rejratio} using \Cref{lem:rejratiointer} and \Cref{lem:lancz}}
\rejratiolem*
\begin{proof}[of \Cref{lem:rejratio}]
    Now we complete the proof of \Cref{lem:rejratio}.
    The rejection ratio $r_k$ is defined as $\frac{\bin(k)}{\alpha \dhat(k)}$. Using \Cref{lem:rejratiointer} and \Cref{lem:lancz}, we can derive the following bound on the rejection ratio: $\widetilde{r}_k = \frac{\interdist(k)}{\alpha \dhat(k)} \cdot \frac{\bin(k)}{\interdist(k)} \leq (1 + (1110n + 2540)\beta\eps + 14\eps\log\left( \dhat(k)\right) + 15\zeta + \sOh(\eps))$. Similarly we can lower bound $\widetilde{r}_k$ by $(1 - (1110n + 2540)\beta\eps - 14\eps\log\left( \dhat(k)\right) - 15\zeta - \sOh(\eps))$.
    \qed
\end{proof}

\subsection*{Proof of \Cref{lem:dhatdist}}
{\dhatdistlem*}

\begin{proof}
Without loss of generality, let us assume that $\dhat(-1) = \dhat(n+1) = 0$. 
Let us also assume that the ideal output of $\lfloor\dinvhat(u)\rfloor$ is $k$, and let $k'$ be the value of $\lfloor{\dinvhat}(u)\rfloor$ as computed in \Cref{alg:btrapprox}. Note that $k' \in (1\pm \epsh) \dinvhat(u)$. Now, because, $\beta \geq 2 \lceil\log_2 n\rceil$, we have $\epsh\dinvhat(u) \leq c\eps n \leq 1$, and hence, $k' \in \{k-1, k, k+1\}$. Since for all $t \in [n]$,  $\Pr\limits_{u}\left(\lfloor\dinvhat(u)\rfloor = t\right) = \dhat(t)$, 
\begin{align*}
    \widetilde{\dhat}(k) &= \sum_{t \in \{k-1, k, k+1\}}\Pr_{\eps,u}\left(k' = k \mid \lfloor\dinvhat(u)\rfloor = t\right) \cdot \Pr_u\left(\lfloor\dinvhat(u)\rfloor = t\right)\\
    &= \sum_{t \in \{k-1, k, k+1\}}\Pr_{\eps,u}\left(k' = k \mid \lfloor\dinvhat(u)\rfloor = t\right) \cdot \dhat(t)
\end{align*}
Observe that, given $\lfloor\dinvhat(u)\rfloor = \crdy{t}$, $k' = \crdy{t}+1$ is possible only when $(1+\epsh)\dinvhat(u) \geq \crdy{t}+1$, that is, $\dinvhat(u) \geq \frac{\crdy{t}+1}{(1+\epsh)}$.
Therefore, given $\lfloor\dinvhat(u)\rfloor = \crdy{t}$, and assuming $\dinvhat(u)$ is uniformly distributed over the range $[\crdy{t}, \crdy{t}+1)$,
we have, 
\begin{align*}
    &\Pr_{\eps,u}\left(k' = \crdy{t}+1 \mid \lfloor\dinvhat(u)\rfloor = \crdy{t}\right) \\
    =& \int_{v=\crdy{t}}^{v=\crdy{t}+1}\Pr_{\eps}\left(k' = \crdy{t}+1 \mid \lfloor\dinvhat(u)\rfloor = \crdy{t} , \dinvhat(u) = v\right) \cdot \\
    &\qquad\qquad\qquad\qquad\qquad f\left(\dinvhat(u) =v \mid \lfloor\dinvhat(u)\rfloor = \crdy{t}\right) \,dv\\
    =&  \int_{v=\frac{\crdy{t}+1}{1+\epsh}}^{v=\crdy{t}+1}\Pr_{\eps}\left(k' = \crdy{t}+1 \mid \lfloor\dinvhat(u)\rfloor = \crdy{t} , \dinvhat(u) = v\right) \,dv
\end{align*}
where $f$ is the uniform probability density function.
The last equality follows from the fact that, we have $\Pr_{\eps}\left(k' = \crdy{t}+1 \mid \lfloor\dinvhat(u)\rfloor = t , \dinvhat(u) = v\right) = 0$ for $v \in \left[\frac{\crdy{t}}{1-\epsh}, \frac{\crdy{t}+1}{1+\epsh}\right]$. 

Since the error in the computation of $\dinvhat(u)$ is uniformly distributed over the range $[-\epsh, \epsh]$, therefore, 
$\Pr_{\eps}\left(k' = \crdy{t}+1 \mid \lfloor\dinvhat(u)\rfloor = \crdy{t} , \dinvhat(u) = v\right) \leq 1/2$ for all $v \in \left[\frac{\crdy{t}+1}{1+\epsh}, \crdy{t}+1\right]$. 
Thus, we have 
$\Pr\limits_{\eps,u}\left(k' = \crdy{t}+1 \mid \lfloor\dinvhat(u)\rfloor = \crdy{t}\right) \leq \int_{v=\frac{\crdy{t}+1}{1+\epsh}}^{v=\crdy{t}+1} \frac{1}{2} \, dv = \frac{\epsh(\crdy{t}+1)}{2(1+\epsh)} \leq \frac{\epsh(\crdy{t}+1)}{3}$.
Similarly, we can upper bound the probability of $k'$ being $\crdy{t}-1$ given $\lfloor\dinvhat(u)\rfloor = \crdy{t}$, by $\Pr\limits_{\eps,u}\left(k' = \crdy{t}-1 \mid \lfloor\dinvhat(u)\rfloor = \crdy{t}\right) \leq \frac{\epsh (\crdy{t}+1)}{3}$.
Therefore, combining, we can upper bound the probability mass:
\begin{align*}
    \label{eq:hub}
    \widetilde{\dhat}(k) &\leq \frac{\epsh k}{3}  \dhat(k-1) + \dhat(k) + \frac{\epsh (k+2)}{3}  \dhat(k+1) \notag \leq \left(1+w_k\epsh (k+2)\right) \dhat(k)
\end{align*}
Next, observing that $\dinvhat(u) \in \left[\frac{k}{1-\epsh}, \frac{k+1}{1+\epsh}\right]$ ensures no possibility of error, we derive a lower bound for the probability of $k'=k$ given $\lfloor\dinvhat(u)\rfloor = k$ as 
\begin{align*}
    \Pr_{\eps,u}\left(k' = k \mid \lfloor\dinvhat(u)\rfloor = k\right) 
    &= \int_{v=k}^{v=k+1}\Pr_{\eps,u}\left(k' = k \mid \lfloor\dinvhat(u)\rfloor = k , \dinvhat(u) = v\right) \cdot\\
    &\qquad\qquad\qquad\quad f\left(\dinvhat(u) =v \mid \lfloor\dinvhat(u)\rfloor = k\right) \, dv\\
    \geq  \int_{v=\frac{k}{1-\epsh}}^{v=\frac{k+1}{1+\epsh}}&\Pr_{\eps,u}\left(k' = k \mid \lfloor\dinvhat(u)\rfloor = k , \dinvhat(u) = v\right) \, dv
\end{align*}
Since for all $v \in \left[\frac{k}{1-\epsh}, \frac{k+1}{1+\epsh}\right]$, $\Pr_{\eps,u}\left[k' = k \mid \lfloor\dinvhat(u)\rfloor = k , \dinvhat(u) = v\right] = 1$,
\begin{align*}
    \Pr_{\eps,u}\left(k' = k \mid \lfloor\dinvhat(u)\rfloor = k\right) \geq \frac{k+1}{(1+\epsh)} - \frac{k}{(1-\epsh)} \geq (1-\epsh(3k+1))
\end{align*}
Where the last inequality follows from the fact that $\frac{1}{1+\epsh} \geq 1 - \epsh$ and $\frac{1}{1-\epsh} \leq 1 + 2\epsh$ for $\epsh \leq \frac{1}{2}$. Consequently, $ \widetilde{\dhat}(k)$
 is at least 
 $\Pr_{\eps,u}\left(k' = k \mid \lfloor\dinvhat(u)\rfloor = k\right) \cdot \dhat(k) \notag = (1-\epsh(3k+1)) \dhat(k) $.
This completes the proof. \qed
\end{proof}


\begin{thebibliography}{10}
\providecommand{\url}[1]{\texttt{#1}}
\providecommand{\urlprefix}{URL }
\providecommand{\doi}[1]{https://doi.org/#1}

\bibitem{arora2009computational}
Arora, S., Barak, B.: Computational complexity: a modern approach. Cambridge
  University Press (2009)

\bibitem{pmlr-v206-banerjee23a}
Banerjee, A., Chakraborty, S., Chakraborty, S., Meel, K.S., Sarkar, U., Sen,
  S.: Testing of horn samplers. In: AISTATS (2023)

\bibitem{bhattacharyya2024testing}
Bhattacharyya, R., Chakraborty, S., Pote, Y., Sarkar, U., Sen, S.: Testing
  self-reducible samplers. In: AAAI (2024)

\bibitem{binder1992monte}
Binder, K., Heermann, D.W., Binder, K.: Monte Carlo simulation in statistical
  physics, vol.~8. Springer (1992)

\bibitem{blanchard2021accurately}
Blanchard, P., Higham, D.J., Higham, N.J.: Accurately computing the log-sum-exp
  and softmax functions. IMA Journal of Numerical Analysis  \textbf{41}(4),
  2311--2330 (2021)

\bibitem{bloom1970space}
Bloom, B.H.: Space/time trade-offs in hash coding with allowable errors.
  Communications of the ACM  \textbf{13}(7),  422--426 (1970)

\bibitem{bonnot2023formally}
Bonnot, P., Boyer, B., Faissole, F., March{\'e}, C., Rieu-Helft, R.: Formally
  verified bounds on rounding errors in concrete implementations of
  logarithm-sum-exponential functions. Tech. rep. (2023)

\bibitem{bonnot2024formally}
Bonnot, P., Boyer, B., Faissole, F., March{\'e}, C., Rieu-Helft, R.: Formally
  verified rounding errors of the logarithm sum exponential function. In:
  FMCAD. pp. 251--260. TU Wien Academic Press (2024)

\bibitem{borwein1984arithmetic}
Borwein, J.M., Borwein, P.B.: The arithmetic-geometric mean and fast
  computation of elementary functions. SIAM review  \textbf{26}(3),  351--366
  (1984)

\bibitem{zimmermann2006error}
Brent, R., Percival, C., Zimmermann, P.: Error bounds on complex floating-point
  multiplication. Mathematics of Computation  \textbf{76}(259),  1469--1481
  (2007)

\bibitem{carter1977universal}
Carter, J.L., Wegman, M.N.: Universal classes of hash functions. In: STOC
  (1977)

\bibitem{chakraborty2019testing}
Chakraborty, S., Meel, K.S.: On testing of uniform samplers. In: {AAAI} (2019)

\bibitem{chakraborty2016algorithmic}
Chakraborty, S., Meel, K.S., Vardi, M.Y.: Algorithmic improvements in
  approximate counting for probabilistic inference: From linear to logarithmic
  sat calls. In: IJCAI. pp. 3569--3576 (2016)

\bibitem{devroye1980generating}
Devroye, L.: Generating the maximum of independent identically distributed
  random variables. Computers \& Mathematics with Applications  \textbf{6}(3),
  305--315 (1980)

\bibitem{devroye2006nonuniform}
Devroye, L.: Nonuniform random sample generation. Handbooks in operations
  research and management science  \textbf{13},  83--121 (2006)

\bibitem{farach2015exact}
Farach-Colton, M., Tsai, M.T.: Exact sublinear binomial sampling. Algorithmica
  \textbf{73},  637--651 (2015)

\bibitem{fousse2007mpfr}
Fousse, L., Hanrot, G., Lef{\`e}vre, V., P{\'e}lissier, P., Zimmermann, P.:
  Mpfr: A multiple-precision binary floating-point library with correct
  rounding. ACM Transactions on Mathematical Software (TOMS)  \textbf{33}(2),
  13--es (2007)

\bibitem{galassi2002gnu}
Galassi, M., Davies, J., Theiler, J., Gough, B., Jungman, G., Alken, P., Booth,
  M., Rossi, F., Ulerich, R.: GNU scientific library. Network Theory Limited
  Godalming (2002)

\bibitem{harris2020array}
Harris, C.R., Millman, K.J., van~der Walt, S.J., Gommers, R., Virtanen, P.,
  Cournapeau, D., Wieser, E., Taylor, J., Berg, S., Smith, N.J., Kern, R.,
  Picus, M., Hoyer, S., van Kerkwijk, M.H., Brett, M., Haldane, A., del
  R{\'{i}}o, J.F., Wiebe, M., Peterson, P., G{\'{e}}rard-Marchant, P.,
  Sheppard, K., Reddy, T., Weckesser, W., Abbasi, H., Gohlke, C., Oliphant,
  T.E.: Array programming with {NumPy}. Nature  \textbf{585},  357--362 (2020)

\bibitem{hoare1961algorithm1}
Hoare, C.A.R.: Algorithm 64: quicksort. Communications of the ACM
  \textbf{4}(7), ~321 (1961)

\bibitem{hopcroft2001introduction}
Hopcroft, J.E., Motwani, R., Ullman, J.D.: Introduction to automata theory,
  languages, and computation. Acm Sigact News  \textbf{32}(1),  60--65 (2001)

\bibitem{hormann1993generation}
H{\"o}rmann, W.: The generation of binomial random samples. Journal of
  statistical computation and simulation  \textbf{46}(1-2),  101--110 (1993)

\bibitem{hormann2004transformed}
H{\"o}rmann, W., Leydold, J., Derflinger, G., H{\"o}rmann, W., Leydold, J.,
  Derflinger, G.: Transformed density rejection (tdr). Automatic Nonuniform
  Random Variate Generation pp. 55--111 (2004)

\bibitem{jeannerod2017error}
Jeannerod, C.P., Kornerup, P., Louvet, N., Muller, J.M.: Error bounds on
  complex floating-point multiplication with an fma. Mathematics of Computation
   \textbf{86}(304),  881--898 (2017)

\bibitem{jeannerod2013improved}
Jeannerod, C.P., Rump, S.M.: Improved error bounds for inner products in
  floating-point arithmetic. SIAM Journal on Matrix Analysis and Applications
  \textbf{34}(2),  338--344 (2013)

\bibitem{jeannerod2018relative}
Jeannerod, C.P., Rump, S.M.: On relative errors of floating-point operations:
  optimal bounds and applications. Mathematics of computation
  \textbf{87}(310),  803--819 (2018)

\bibitem{kachitvichyanukul1988Binomial}
Kachitvichyanukul, V., Schmeiser, B.W.: Binomial random sample generation.
  Communications of the ACM  \textbf{31}(2),  216--222 (1988)

\bibitem{karney2016sampling}
Karney, C.F.: Sampling exactly from the normal distribution. ACM Transactions
  on Mathematical Software (TOMS)  \textbf{42}(1),  1--14 (2016)

\bibitem{karp1983monte}
Karp, R.M., Luby, M.: Monte-carlo algorithms for enumeration and reliability
  problems. In: FOCS (1983)

\bibitem{karp1989monte}
Karp, R.M., Luby, M., Madras, N.: Monte-carlo approximation algorithms for
  enumeration problems. Journal of algorithms  \textbf{10}(3),  429--448 (1989)

\bibitem{kumar2023tolerant}
Kumar, G., Meel, K.S., Pote, Y.: Tolerant testing of high-dimensional samplers
  with subcube conditioning (2023)

\bibitem{lanczos1964precision}
Lanczos, C.: A precision approximation of the gamma function. Journal of the
  Society for Industrial and Applied Mathematics, Series B: Numerical Analysis
  \textbf{1}(1),  86--96 (1964)

\bibitem{meduna2000turing}
Meduna, A., Meduna, A.: Turing transducers. Automata and Languages: Theory and
  Applications pp. 833--887 (2000)

\bibitem{meel2020testing}
Meel, K.S., Pote, Y.P., Chakraborty, S.: On testing of samplers. NeurIPS
  (2020)

\bibitem{meel2019not}
Meel, K.S., Shrotri, A.A., Vardi, M.Y.: Not all fprass are equal: demystifying
  fprass for dnf-counting. Constraints  \textbf{24},  211--233 (2019)

\bibitem{meel2021estimating}
Meel~rG, K.S., Vinodchandran~rG, N., Chakraborty, S.: Estimating size of the
  union of sets in streaming model  (2021)

\bibitem{muller2006elementary}
Muller, J.M., Muller, J.M.: Elementary functions. Springer (2006)

\bibitem{ormann1994transformed}
ormann, W., erflinger, G.: The transformed rejection method for generating
  random variables, an alternative to the ratio of uniforms method.
  Communications in Statistics-Simulation and Computation  \textbf{23}(3),
  847--860 (1994)

\bibitem{PM22}
Pote, Y., Meel, K.S.: On scalable testing of samplers. NeurIPS  (2022)

\bibitem{PM21}
Pote, Y.P., Meel, K.S.: Testing probabilistic circuits. NeurIPS  (2021)

\bibitem{pugh2004analysis}
Pugh, G.R.: An analysis of the Lanczos gamma approximation. Ph.D. thesis,
  University of British Columbia (2004)

\bibitem{pugh1990concurrent}
Pugh, W.: Concurrent maintenance of skip lists. Citeseer (1990)

\bibitem{rump2008accurate}
Rump, S.M., Ogita, T., Oishi, S.: Accurate floating-point summation part i:
  Faithful rounding. SIAM Journal on Scientific Computing  \textbf{31}(1),
  189--224 (2008)

\bibitem{schmeiser1981poisson}
Schmeiser, B., Kachitvichyanukul, V.: Poisson random sample generation.
  Research memorandum pp. 81--4 (1981)

\bibitem{sharma2019ganak}
Sharma, S., Roy, S., Soos, M., Meel, K.S.: Ganak: A scalable probabilistic
  exact model counter. In: IJCAI. vol.~19, pp. 1169--1176 (2019)

\bibitem{soosengineering}
Soos, M., Aggarwal, D., Chakraborty, S., Meel, K.S., Obremski, M.: Engineering
  an efficient approximate dnf-counter. In: IJCAI. pp. 2031--2038 (2023)

\bibitem{soos2019bird}
Soos, M., Meel, K.S.: Bird: engineering an efficient cnf-xor sat solver and its
  applications to approximate model counting. In: Proceedings of the AAAI
  Conference on Artificial Intelligence. vol.~33, pp. 1592--1599 (2019)

\bibitem{mpfralgorithms}
{The MPFR Team}: The mpfr library: Algorithms and proofs,
  \url{https://www.mpfr.org/algorithms.pdf/}

\bibitem{thomopoulos2012essentials}
Thomopoulos, N.T.: Essentials of Monte Carlo simulation: Statistical methods
  for building simulation models. Springer Science \& Business Media (2012)

\end{thebibliography}
\end{document}